\documentclass[a4paper,UKenglish,cleveref,autoref,thm-restate,authorcolumns,pdfa,numberwithinsect]{lipics-v2021}
\pdfoutput=1 %uncomment to ensure pdflatex processing (mandatatory e.g. to submit to arXiv)
\hideLIPIcs{} %uncomment to remove references to LIPIcs series (logo, DOI, ...), e.g. when preparing a pre-final version to be uploaded to arXiv or another public repository

\title{How to Reduce Temporal Cliques to Find Sparse Spanners} %TODO Please add

\author{Sebastian Angrick}{University of Potsdam, Hasso Plattner  Institute, Faculty of Digital Engineering}{sebastian.angrick@student.hpi.de}{https://orcid.org/0009-0007-9840-9611}{}
\author{Ben Bals}{University of Potsdam, Hasso Plattner  Institute, Faculty of Digital Engineering}{ben.bals@student.hpi.de}{}{}
\author{Tobias Friedrich}{University of Potsdam, Hasso Plattner  Institute, Faculty of Digital Engineering}{tobias.friedrich@hpi.de}{}{}
\author{Hans {\defaultGraphName}awendowicz}{University of Potsdam, Hasso Plattner  Institute, Faculty of Digital Engineering}{hans.gawendowicz@hpi.de}{https://orcid.org/0000-0001-8394-4469}{}
\author{Niko Hastrich}{University of Potsdam, Hasso Plattner  Institute, Faculty of Digital Engineering}{niko.hastrich@student.hpi.de}{https://orcid.org/0000-0002-3753-145X}{}
\author{Nicolas Klodt}{University of Potsdam, Hasso Plattner  Institute, Faculty of Digital Engineering}{nicolas.klodt@hpi.de}{https://orcid.org/0000-0002-5406-7441}{}
\author{Pascal Lenzner}{University of Potsdam, Hasso Plattner  Institute, Faculty of Digital Engineering}{pascal.lenzner@hpi.de}{}{}
\author{Jonas Schmidt}{University of Potsdam, Hasso Plattner  Institute, Faculty of Digital Engineering}{jonas.schmidt@student.hpi.de}{https://orcid.org/0000-0002-1115-3868}{}
\author{George Skretas}{University of Potsdam, Hasso Plattner  Institute, Faculty of Digital Engineering}{georgios.skretas@hpi.de}{}{}
\author{Armin Wells}{University of Potsdam, Hasso Plattner  Institute, Faculty of Digital Engineering}{armin.wells@student.hpi.de}{https://orcid.org/0000-0002-5459-7899}{}

\authorrunning{Angrick et al.} %TODO mandatory. First: Use abbreviated first/middle names. Second (only in severe cases): Use first author plus 'et al.'

\Copyright{Sebastian Angrick, Ben Bals, Hans Gawendowicz, Niko Hastrich, Nicolas Klodt, Pascal Lenzer, Jonas Schmidt, George Skretas, and Armin Wells} %TODO mandatory, please use full first names. LIPIcs license is "CC-BY";  http://creativecommons.org/licenses/by/3.0/


 %TODO mandatory; please add comma-separated list of keywords

\category{} %optional, e.g. invited paper

\relatedversion{} %optional, e.g. full version hosted on arXiv, HAL, or other respository/website
\nolinenumbers %uncomment to disable line numbering

\EventEditors{John Q. Open and Joan R. Access}
\EventNoEds{2}
\EventLongTitle{42nd Conference on Very Important Topics (CVIT 2016)}
\EventShortTitle{CVIT 2016}
\EventAcronym{CVIT}
\EventYear{2016}
\EventDate{December 24--27, 2016}
\EventLocation{Little Whinging, United Kingdom}
\EventLogo{}
\SeriesVolume{42}
\ArticleNo{23}
\usepackage{caption}
\usepackage{subcaption}

\usepackage{mathtools}
\usepackage[ruled]{algorithm2e}
\usepackage[skins]{tcolorbox}

\usepackage{booktabs}

\DeclareMathOperator{\bigO}{\mathcal{O}}
\DeclareMathOperator{\smallO}{o}
\DeclareMathOperator{\N}{\mathbb{N}}
\DeclareMathOperator{\Z}{\mathbb{Z}}

\DeclareMathOperator*{\argmin}{arg\,min}
\DeclareMathOperator*{\argmax}{arg\,max}

\let\originalleft\left
\let\originalright\right
\renewcommand{\left}{\mathopen{}\mathclose\bgroup\originalleft}
\renewcommand{\right}{\aftergroup\egroup\originalright}

\newcommand{\biclique}[3]{\left({#1}, {#2}, {#3}\right)}

\newcommand{\cliqueEdges}[1]{{#1} \otimes {#1}}
\newcommand{\bicliqueEdges}[2]{{#1} \otimes {#2}}

\newcommand{\undirEdge}[2]{\left\{{#1}, {#2}\right\}}
\DeclareMathOperator{\edgeLabel}{\lambda}
\newcommand{\edgeLabele}[1]{\edgeLabel\left({#1}\right)}
\newcommand{\edgeLabelv}[2]{\edgeLabele{\undirEdge{{#1}}{{#2}}}}
\DeclareMathOperator{\cliqueSpannerSizeOp}{\mathcal{C}}
\DeclareMathOperator{\bicliqueBispannerSizeOp}{\mathcal{D}}
\newcommand{\cliqueSpannerSize}[1]{\cliqueSpannerSizeOp\left({#1}\right)}
\newcommand{\bicliqueBispannerSize}[1]{\bicliqueBispannerSizeOp\left({#1}\right)}
\newcommand{\defaultBiclique}{D = \biclique{A}{B}{\edgeLabel}}
\newcommand{\defaultGraphName}{G}

\DeclareMathOperator{\earlyMatching}{\pi^{\text{--}}}
\newcommand{\earlyPartner}[1]{\earlyMatching\left({#1}\right)}
\DeclareMathOperator{\lateMatching}{\pi^+}
\newcommand{\latePartner}[1]{\lateMatching\left({#1}\right)}
\newcommand{\lateMatchingG}[1]{\pi^+_{{#1}}}

\newcommand{\lessOp}[1]{\prec_{#1}}

\newcommand{\leqAt}[3]{{#1} \preceq_{#3} {#2}}
\newcommand{\geqAt}[3]{{#1} \succeq_{#3} {#2}}
\newcommand{\lessAt}[3]{{#1} \prec_{#3} {#2}}
\newcommand{\greaterAt}[3]{{#1} \succ_{#3} {#2}}

\newcommand{\orderingAt}[1]{\sigma_{#1}}
\newcommand{\edgeIndex}[2]{\orderingAt{#1}\left({#2}\right)}

\DeclareMathOperator{\inSetOp}{In}
\DeclareMathOperator{\outSetOp}{Out}
\newcommand{\inSet}[1]{\inSetOp\left({#1}\right)}
\newcommand{\outSet}[1]{\outSetOp\left({#1}\right)}
\newcommand{\inSetTime}[2]{\inSetOp_{#1}^{#2}}
\newcommand{\outSetTime}[2]{\outSetOp_{#1}^{#2}}

\DeclareMathOperator{\ringShiftOp}{SM}
\newcommand{\ringShift}[1]{\ringShiftOp\left({#1}\right)}
\DeclareMathOperator{\xmod}{mod}
\newcommand{\setInc}[1]{[{#1}]}
\newcommand{\setExc}[1]{\{0, \ldots, {#1}-1\}}

\DeclareMathOperator{\maliciousOp}{NotRev}
\newcommand{\malicious}[1]{\maliciousOp_{{#1}}}

\newif\iflong
\newif\ifshort

\longtrue

\iflong
\else
\shorttrue
\fi

\usepackage{xparse}

\DeclareDocumentCommand{\set}{m g o}%
{%
    \IfNoValueTF{#3}{\left}{#3}\{#1
            \IfNoValueTF{#2}{}{\ \IfNoValueTF{#3}{\left}{#3}\vert\ \vphantom{#1}#2\IfNoValueTF{#3}{\right.}{}}
                \IfNoValueTF{#3}{\right}{#3}\}%
}

\DeclareDocumentCommand{\abs}{m o}%
{%
    \IfNoValueTF{#2}{\left}{#2}\vert#1
                \IfNoValueTF{#2}{\right}{#2}\vert%
}

\newcommand{\caseIf}{\text{if }}

\newcommand*{\ie}[1]{(i.e., #1)}
\newcommand*{\eg}[1]{(e.g., #1)}
\DeclareMathOperator{\exGraphOp}{\times}
\newcommand{\exGraph}[2]{{#1} \exGraphOp {#2}}
\DeclareMathOperator{\bagMapOp}{\phi}
\newcommand{\bagMap}[1]{\bagMapOp\left({#1}\right)}

\newcommand{\smsmbg}[2]{\ringShiftOp(#1, #2)}

\newcommand{\nodes}{vertices}

\begin{document}

\maketitle

\begin{abstract}
	Many real-world networks, such as transportation or trade networks, are
dynamic in the sense that the edge-set may change over time, but these changes
are known in advance. This behavior is captured by the temporal graphs model,
which has recently become a trending topic in theoretical computer science. A
core open problem in the field is to prove the existence of linear-size temporal
spanners in temporal cliques, i.e., sparse subgraphs  of complete temporal
graphs that ensure all-pairs reachability via temporal paths. So far, the best
known result is the existence of temporal spanners with $\bigO(n\log n)$ many edges.
We present significant progress towards proving whether linear-size temporal
spanners exist in all temporal cliques.

We adapt techniques used in previous works and heavily expand and generalize
them. This allows us to show that the
existence of a linear spanner in cliques and bi-cliques is equivalent and using this, we provide a simpler and more intuitive proof of the $\bigO(n\log n)$ bound by giving an efficient algorithm for finding linearithmic spanners.
Moreover, we use our novel and efficiently computable approach to show that a large class of temporal
cliques, called edge-pivotable graphs, admit linear-size temporal spanners. To contrast this, we investigate other
classes of temporal cliques that do not belong to the class of edge-pivotable
graphs. We introduce two such graph classes and we develop novel algorithmic techniques for
establishing the existence of linear temporal spanners in these graph classes as
well.

% Keep this at the end of each file
%%% Local Variables:
%%% mode: latex
%%% TeX-master: "../writeup.tex"
%%% End:

\end{abstract}

%\vfill\pagebreak

\newcommand{\george}[1]{{\color{purple}[George: #1]}}

\newcommand{\pascal}[1]{{\color{blue}[Pascal: #1]}}

\newpage
\section{Introduction}
Many real-world networks, like transportation networks with scheduled train connections or social networks with repeating meeting schedules, are dynamic, i.e., their edge-set can change over time. To address this, \emph{temporal graphs} have gained significant attention in recent research in theoretical computer science ~\cite{kempe_connectivity,Michail2015}. In these graphs it is assumed that the edge-set can change in every time step but the schedule of the availability of edges is known in advance. 
Various problems have been studied in this model, including network redesign \cite{DBLP:conf/ijcai/DeligkasES23,DBLP:journals/jcss/EnrightMMZ21,DBLP:conf/stacs/FuchsleMNR22}, vertex cover \cite{AkridaTemporalVertex2020,DondiTimeline2023,HammTheComplexity2022}, and influence maximization \cite{DBLP:conf/atal/DeligkasEGS23}. A general theme is that insights from the static graph setting are only of limited value in temporal graphs, since the temporal availability of the edges makes these problems much harder or even infeasible to solve. For example, even fundamental statements like Menger's Theorem do not hold for temporal graphs~\cite{MertziosTemporalNetwork2013}. Also, one of the most essential problems in graphs, that of finding a sparse spanner, is much harder in temporal graphs. A sparse spanner is a small subset of the edge-set that ensures all-pairs connectivity. In static graphs, it is well-known that linear-size spanners exist and can be computed efficiently, e.g., minimum spanning trees~\cite{cormen2022introduction}.

Finding sparse \emph{temporal spanners} is much more intricate in temporal graphs, since there paths do not necessarily compose \ie{concatenating two paths where the endpoint of the first is the starting point of the second does not necessarily yield a temporal path}.
More specifically, Kempe, Kleinberg, and
Kumar~\cite{kempe_connectivity} showed that there is a class of temporal graphs
having \(\Omega(n \log n)\) many edges, such that no single edge can be removed
while preserving temporal connectivity. In the same work, the authors asked,
what is the minimum number of edges \(c(n)\), such that any temporal graph with
\(n\) vertices has a temporal spanner using at most \(c(n)\) many edges.  Almost
15 years later, this question was answered by Axiotis and
Fotakis~\cite{axiotis_size_approx_spanner} by providing a class of dense but non-complete
temporal graphs having \(\Theta(n^{2})\) edges, such that any subgraph
preserving temporal connectivity necessarily needs to contain \(\Omega(n^{2})\)
edges as well.

This negative result raises the question if there even exists a natural class of temporal graphs that
can be sparsified without losing temporal connectivity. In particular, the question if every temporal clique has a spanner of linear size has
remained open for over 20 years now. In a 2019 breakthrough paper, Casteigts, Peters, and Schoeters~\cite{casteigts_fireworks_conf} proved that temporal cliques always admit a
spanner of size \(O(n\log n)\). Not much progress has been made since then, with most work on the topic veering off towards the study of temporal spanners in random temporal graphs \cite{casteigts_threshold}, or focusing on low stretch \cite{BiloBlackout2022,BiloSparse2022} or a game-theoretic setting~\cite{BiloTemporalNetwork2023}. 

In this paper, we present significant progress on the long-standing open problem by showing that a large class of temporal cliques does admit linear-size temporal spanners.

\subsection{Our Contribution}
In the first part of the paper, we adapt techniques used in previous works and heavily expand and generalize them. In particular, in~\cite{casteigts_fireworks_conf} a temporal spanner in cliques is computed by a reduction from a clique to a bi-clique. We deepen this connection 
by showing that the existence of a linear spanner in cliques and bi-cliques actually is equivalent.
Furthermore, in the reduction of Casteigts, Peters, and Schoeters~\cite{casteigts_fireworks_conf}, the authors exploit the property that the set containing the earliest edge of each vertex is a perfect matching of the bi-clique and that this also holds for the set of latest edges.  
We show
that this property can also be assumed when searching for bi-spanners in bi-cliques directly. This yields a simple and more intuitive algorithm for computing linearithmic spanners, thus proving the existence of size $O(n \log n)$ temporal spanners in temporal cliques.

Moreover, we reconsider the notion of a \emph{pivot-vertex}~\cite{casteigts_fireworks_conf,casteigts_threshold}. This is a vertex $u$  that can be reached by every vertex in the graph until some time step $t$, and then $u$ can reach every other vertex in the graph by a temporal path starting at or after time $t$. 
We transfer this notion to bi-cliques, extend it to \emph{$c$-pivot-edges}, and we strengthen it significantly to find small spanners. If a bi-clique contains a $c$-pivot-edge, we can iteratively reduce the size of the graph, while maintaining a linear-size temporal spanner, until the instance is solved, or the instance no longer has a $c$-pivot-edge.
This reduction rule is widely applicable and yields linear-size temporal spanners in many cases.

In the second part of the paper, we investigate temporal cliques that are not edge-pivotable graphs.
We identify two such graph classes, called \emph{shifted matching graphs} and \emph{product graphs}, and for both we provide novel techniques that allow us to efficiently compute linear-size spanners. Thus, we extend the class of temporal cliques that admit linear-size temporal spanners even further. It is open if our classes cover every temporal clique. However,  
finding other graphs that do not belong to our graph classes seems to be challenging.

\ifshort
Statements where proofs or details are omitted due to space constraints are marked with~"$\star$". A full version containing all proofs and details is provided in the appendix. 
\fi

\subsection{Detailed Discussion of Related Work}

The study of temporal spanners was initiated by Kempe, Kleinberg and Kumar \cite{kempe_connectivity}. In their seminal paper, they introduce the temporal graphs model and show that spanners on temporal hypercubes have \(\Omega(n \log n)\) edges. Axiotis and Fotakis \cite{axiotis_size_approx_spanner} improve on this by providing a class of graphs with \(\Theta(n^{2})\) edges, such that no edge can be removed while maintaining temporal connectivity. Later, Casteigts, Peters, and Schoeters~\cite{casteigts_fireworks_conf} study temporal cliques and show that temporal spanners of size \(O(n\log n)\) always exist. Their "fireworks" algorithm combines different techniques, such as dismountability and delegation, to reduce the clique to a bipartite graph, and then compute a spanner on that graph, that can then be translated to a spanner of the original graph.

Casteigts, Raskin, Renken and Zamaraev \cite{casteigts_threshold} take a different approach. %to the problem of computing temporal spanners on a clique. 
They consider random temporal graphs, where the host graph is a clique, and each edge is added to the temporal graph with probability $p$. They show a sharp threshold for the generated temporal graph being temporally connected. Even more surprisingly, they show that as soon as this threshold is reached, the generated graph will contain an almost optimal spanner of size $2n+o(n)$. On a different note, Bilò, D'Angelo, Gualà, Leucci, and Rossi~\cite{BiloSparse2022} investigate whether temporal cliques admit small temporal spanners with low stretch. Their main result shows that they can always compute a temporal spanner with stretch $O(\log n)$ and $O(n \log^2 n)$ edges.

% Keep this at the end of each file
%%% Local Variables:
%%% mode: latex
%%% TeX-master: "../writeup.tex"
%%% End:
% LocalWords:  sparsified rechability iteratively

\section{Preliminaries}\label{sec:prelims}
\newcommand\restr[2]{{% we make the whole thing an ordinary symbol
  \left.\kern-\nulldelimiterspace % automatically resize the bar with \right
  #1 % the function
  \vphantom{\big|} % pretend it's a little taller at normal size
  \right|_{#2} % this is the delimiter
  }}

For two sets $A, B$, we write $\bicliqueEdges{A}{B} \coloneqq \{\undirEdge{a}{b}
\mid a \in A, b \in B, \text{ and } a \ne b\}$. If $A$ and $B$ are disjoint, we
write $A \sqcup B$ for their union. For $n \in \N^+$, let $[n] = \{1, \dots,
n\}$.

Let $(V, E)$ be an undirected graph. This together with a labeling function\footnote{In general each edge could have more than one label. We restrict ourselves to the case of \emph{simple} temporal graphs where each edge has exactly one label. Both formulations are equivalent for our use-case~\cite{casteigts_fireworks_conf}.}
$\edgeLabel \colon E \to \N$ forms the \emph{temporal graph} $\defaultGraphName
= (V, E, \edgeLabel)$. Semantically, an edge $e \in E$ is present at time
$\edgeLabele{e}$. For a vertex-set $U \subseteq V$, write $\defaultGraphName[U]$
for the subgraph of $\defaultGraphName$ induced by $U$. For $v \in V$, let
$N(v)$ denote the set of neighbors of $v$, that is, all vertices~$u$ with
$\undirEdge{v}{u} \in E$. For $v \in V$ and $u, u' \in N(v)$, we define the
order $\leqAt{u}{u'}{v}$, if the edge $\undirEdge{u}{v}$ is not later than the
edge $\undirEdge{u'}{v}$, i.e., if $\edgeLabelv{u}{v} \le \edgeLabelv{u'}{v}$.
We define $\prec_v$, $\succeq_v$, and $\succ_v$ analogously.

If $P = v_1 \dots v_k$ is a path in $(V, E)$ and for all $i \in [k-2]$, we have
$\leqAt{v_{i}}{v_{i+2}}{v_{i+1}}$, then $P$ is a \emph{temporal path} in
$G$.\footnote{This is also referred to as non-strict
temporal paths, where strict temporal paths ask that
$\lessAt{v_i}{v_{i+1}}{v_{i+1}}$.} %Since we construct many paths, 
We will also
write $P = v_1 \to v_2 \overset{S}{\leadsto} v_k$, where, $v_1 \to v_2$ describes
a direct edge and $v_2 \overset{S}{\leadsto} v_k$ is a temporal path
starting after $\edgeLabelv{v_1}{v_2}$ using only edges of $S \subseteq E$.
Since only reachability is our focus, we will also construct non-simple paths. %that visit the same vertex multiple times. 
Note that such paths can be shortcut to
vertex-disjoint paths while remaining temporal.

Vertex $u$ can \emph{reach} vertex $v$ if there is a temporal path from $u$ to $v$. Also, $G$ is called \emph{temporally connected}, if every vertex can reach every other vertex.
A \emph{spanner} for $G$ is a set $S \subseteq E$, such that $(V, S, \restr{\edgeLabel}{S})$ is temporally connected, where $\restr{\edgeLabel}{S}$ is $\edgeLabel$ restricted to~$S$.

The edges with the smallest and largest label take a special
role.
For $v \in V$, let $\earlyPartner{v} \coloneqq \argmin_{u \in N(v)}
\edgeLabelv{v}{u}$ be $v$'s earliest neighbor and let $\latePartner{v} \coloneqq \argmax_{u
\in N(v)} \edgeLabelv{v}{u}$ be $v$'s latest neighbor.
We do not consider graphs with isolated
vertices and we justify in \Cref{sec:cliques-and-bi-cliques} that we only use injective labelings, so this is well-defined. Denote the set of all earliest
edges as $\earlyMatching \coloneqq \{\undirEdge{v}{\earlyPartner{v}}\mid v \in
V\}$ and the set of latest edges as $\lateMatching \coloneqq \{\undirEdge{v}{\latePartner{v}}\mid v \in V\}$. If both sets $\earlyMatching$ and 
$\lateMatching$ form a perfect matching
for $V$, then $G$ is \emph{extremally matched}. Also, for
\(v,w\in V\), and \(e = \undirEdge{v}{w} \in E \), we write \(\edgeIndex{v}{e}\) or simply $\edgeIndex{v}{w}$ for the index of \(e\) in $\setInc{n}$, when ordering all edges incident to
\(v\) ascendingly by their label \eg{if \(w=\earlyPartner{v}\), then
\(\edgeIndex{v}{e}=1\)}.

Let $v \in V$ be a vertex and $t \in \N$ be a time step. We define
the set $\inSetTime{v}{t} \subseteq V$ to be the vertices reaching $v$ via a
temporal path with every label being at most $t$. Respectively, we define
the set $\outSetTime{v}{t} \subseteq V$ to be the vertices that $v$ can reach via any
temporal path starting at or after $t$.  As a shorthand, for an edge $e =
\undirEdge{u}{v} \in E$, we define $\inSet{e} =
\inSetTime{u}{\edgeLabel(e)} \cup \inSetTime{v}{\edgeLabel(e)}$ and $\outSet{e}
= \outSetTime{u}{\edgeLabel(e)} \cup \outSetTime{v}{\edgeLabel(e)}$. Note, that
we always have $u,v \in \inSet{e} \cap \outSet{e}$.

We focus only on temporal graphs where the underlying graph is a clique or a complete bipartite graph. In these cases, the set of edges is self-evident and for cliques we simply omit it and write $C = (V, \edgeLabel)$. We, similarly, denote complete bipartite temporal graphs by a triple $\defaultBiclique$, where $A$ and $B$ are the two parts and $\edgeLabel \colon \bicliqueEdges{A}{B} \to \N$ is the corresponding labeling function. When necessary, we refer to the set of edges as $E(C) \coloneqq \cliqueEdges{V}$ or $E(D) \coloneqq \bicliqueEdges{A}{B}$. We call $C$ a \emph{temporal clique} and $D$ a \emph{temporal bi-clique}. 
Finally, we define an asymmetric analogue of spanners for temporal bi-cliques. We call a set $S \subseteq \bicliqueEdges{A}{B}$ a \emph{bi-spanner} for $D$, if every $a \in A$ can reach every $b \in B$ through a temporal path that is contained in $S$. Note that here, only one side needs to reach the other.

For $e = \undirEdge{u}{v}$ in temporal cliques or bi-cliques, notice that all vertices must be in $\inSet{e}$ or $\outSet{e}$, since all vertices have an edge to $u$ or $v$, which puts this vertex in $\inSet{e}$ or $\outSet{e}$, depending on the time-label. See \Cref{fig:in-out-set} for a visualization of $\inSetOp$ and $\outSetOp$.
\begin{figure}[t]
    \centering
    \begin{subfigure}[t]{0.45\textwidth}
        \centering
        \includegraphics[height=2cm]{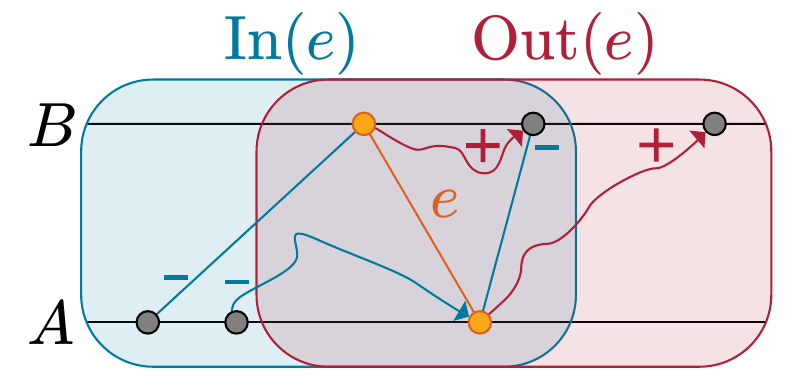}
        \caption{Blue edges and paths have labels earlier than $\edgeLabele{e}$, red ones are later.}
        \label{fig:in-out-set}
    \end{subfigure}
    \hfill
    \begin{subfigure}[t]{0.51\textwidth}
    	\centering
    	\includegraphics[height=2cm]{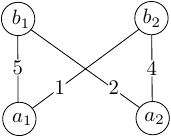}
    	\caption{A bi-clique where $a_1$ cannot reach $a_2$ using a path starting after time $2$ and $a_1$ can not be reached by $a_2$ using a path arriving before $5$, so $\inSetTime{a_1}{3} \cup \outSetTime{a_1}{3} \neq V$.}
    	\label{fig:node-pivot-counterexample}
    \end{subfigure}
    \caption{Examples for the sets $\inSetOp$ and $\outSetOp$.}
    \label{fig:in-out-examples}
\end{figure}

\begin{observation}\label{lem:in-cup-out-eq-all-vertices}
  	Let \(G = (V, E, \edgeLabel)\) be a temporal clique or bi-clique.
	For all $e \in E$, we have $\inSet{e} \cup \outSet{e} = V$.
\end{observation}

Note that for the case of bi-cliques considering both endpoints of $e$ is crucial to ensure that every vertex has a direct edge to an endpoint. Otherwise, there exist simple instances with a vertex $v$ and a time $t$ with $\inSetTime{v}{t} \cup \outSetTime{v}{t} \neq A \sqcup B$. See~\Cref{fig:node-pivot-counterexample}.

\begin{remark}
  For \(v \in A \sqcup B\) and \(t\in \N\), the analogue to
\autoref{lem:in-cup-out-eq-all-vertices}
\ie{\(\inSetTime{v}{t}\cup \outSetTime{v}{t} = A\sqcup B\)} does not necessarily hold, as one can see in \Cref{fig:node-pivot-counterexample}.
\end{remark}

Finally, for $n \in \N$, let $\cliqueSpannerSize{n}$ be the minimum
number such that every temporal clique with $n$ vertices admits a spanner with
at most $\cliqueSpannerSize{n}$ edges. Analogously, let
\(\bicliqueBispannerSize{n}\) be the minimum number such that any bi-clique in which both
sides have \(n\) vertices admits a bi-spanner with at most
\(\bicliqueBispannerSize{n}\) edges. We study the open problem whether $\cliqueSpannerSize{n} \in
\bigO(n)$ and whether $\bicliqueBispannerSize{n} \in
\bigO(n)$ holds.

\section{Cliques and Bi-Cliques}\label{sec:cliques-and-bi-cliques}
In this section, we prove that the asymptotic behavior of spanners in cliques
differs from the asymptotic behavior of bi-spanners in bi-cliques only by a
constant factor \ie{we prove \(\cliqueSpannerSize{n} \in
\Theta(\bicliqueBispannerSize{n})\)}.  First, we show that
\(\cliqueSpannerSize{n} \in \bigO(\bicliqueBispannerSize{n})\). After this, we
show \(\cliqueSpannerSize{n} \in \Omega(\bicliqueBispannerSize{n})\) by proving each inequality of
\(\cliqueSpannerSize{n} \geq \bicliqueBispannerSize{\frac{n}{2}} \geq
\frac{1}{4}\bicliqueBispannerSize{n}\) separately. Bi-cliques and bi-spanners were already
discussed in \cite{casteigts_fireworks_conf}, when deriving \(\cliqueSpannerSize{n} \in
\bigO(n \log n)\). They called these instances \emph{residual} as they remained
once their other techniques were not applicable anymore. Our result can be
seen as evidence that the connection between cliques and bi-cliques is fundamental, beyond being a step in their algorithm. First, we state the main result of this section.

\iflong
\begin{theorem}
\fi
\ifshort
\begin{theorem}
\fi
\label{thm:bi-cliques-are-cliques}
  The worst-case minimum size of a spanner for temporal cliques and temporal bi-cliques only differs by a constant factor, 
  that is, we have \(\cliqueSpannerSize{n} \in \Theta(\bicliqueBispannerSize{n})\).
\end{theorem}

We start with a technical lemma that allows us to focus on globally injective labeling functions for the rest of the paper. This slightly extends a result by \cite{casteigts_fireworks_conf} which proved the generality of locally injective labelings.
Note that the statement holds for all temporal graphs. %, in particular for cliques and bi-cliques.

\iflong
\begin{lemma}
\fi
\ifshort
\begin{lemma}[$\star$]
\fi
\label{lem:injective-labeling}
  Let $G=(V, E, \lambda)$ be a temporal graph. Then there is an injective labeling function $\lambda'$ s.t.\ every spanner $S \subseteq E$ for $(V, E, \lambda')$ is also a spanner for the original graph~$G$. In case the original graph $\defaultBiclique$ is a bi-clique, the same holds for a bi-spanner $S$.
\end{lemma}

\iflong
\begin{proof}
  Assume $\lambda$ is not injective.
  To fix that, we simply arbitrarily order edges with the same label.
  For a simple construction, assume $f \colon E \to [|E|]$ is an arbitrary numbering of the edges and for all $e \in E$ set $\lambda'(e) \coloneqq \lambda(e) \cdot (|E|+1) + f(e)$.
  To see that $\lambda'$ is injective, first take distinct $e_1, e_2 \in E$ with $\lambda(e_1) = \lambda(e_2)$.
  In this case, $\lambda'(e_1) - \lambda'(e_2) = (\lambda(e_1) + f(e_1)) - (\lambda(e_2) + f(e_2)) = f(e_1) - f(e_2)$ and since $f$ is injective this difference is not 0, thus we have $\lambda'(e_1) \ne \lambda'(e_2)$. Now assume we have $e_1, e_2 \in E$ with $\lambda(e_1) < \lambda(e_2)$.
  Then we have 
  \begin{align*}
    \lambda'(e_1) &= \lambda(e_1) \cdot (|E| + 1) + f(e) \\
    &\le \lambda(e_1) \cdot (|E| + 1) + |E| \\
    &< (\lambda(e_1) + 1) \cdot (|E| + 1) \\
    &\le \lambda'(e_2).
  \end{align*}
  Also observe that for all $e_1, e_2 \in E$ we have $\lambda'(e_1) \le \lambda'(e_2)$ then $\lambda(e_1) \le \lambda(e_2)$.
  Thus if $P$ is a valid temporal path w.r.t. $(V, E, \lambda')$, then $P$ is also a valid temporal path w.r.t. $G$.
  Conclude that if $S$ is a spanner in $(V, E, \lambda')$, then $S$ is a spanner in $G$. 
  
  Notice that the same construction works for bi-spanners, since we can again break ties arbitrarily and all remaining temporal paths existed in the original instance.
\end{proof}
\fi

We can now relate minimum spanners in cliques to minimum bi-spanners in bi-cliques.

\iflong
\begin{lemma}
\fi
\ifshort
\begin{lemma}[$\star$]
\fi
\label{lem:clique-spanner-by-bispanner}
  Let \(C=(V,\edgeLabel)\) be a temporal clique. There is a temporal bi-clique
  \(D=(A, B, \edgeLabel')\) with \(|V| = |A| = |B|\) s.t.\ for any bi-spanner
  \(S_D\) in \(D\) there is a spanner \(S_C\) in \(C\) with \(|S_C| \leq |S_D|\).
\end{lemma}

\iflong
\begin{proof}
Remember, by \cref{lem:injective-labeling}, we can assume that \(\edgeLabel\) is an
injection.

For every vertex \(v\in V\), we create a duplicate vertex \(v'\). Call the set
of all duplicates \(V'\). 
We define the edge label between two vertices as the label between the original vertices. For the edge label between a vertex in $V$ and its copy, we choose a label higher than all in the original graph. That is, to define \(\edgeLabel'\), choose \(\mu >
\max_{e \in \cliqueEdges{V}}\edgeLabel(e)\) and set for every \(u,v\in V\) \[
  \edgeLabel'(\undirEdge{u}{v'}) \coloneqq \begin{cases}
    \edgeLabelv{u}{v}, & \text{if } u \neq v;\\
                               \mu,&\caseIf u=v.
                     \end{cases}
\]

An example of this transformation can be seen in \Cref{fig:clique-to-biclique-transform}.
\begin{figure}[t]
  \captionsetup[subfigure]{justification=centering}
  \centering
  \begin{subfigure}{0.49\textwidth}
    \centering
    \includegraphics[height=3cm]{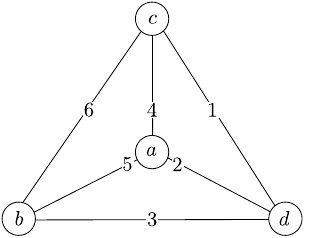}
    \caption{A temporal clique \(C\).}
  \end{subfigure}
  \begin{subfigure}{0.49\textwidth}
    \centering
    \includegraphics[height=3cm]{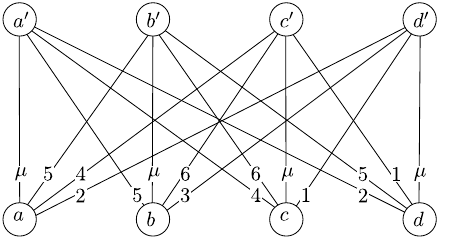}
    \caption{The corresponding temporal bi-clique \(D\).}
  \end{subfigure}
  \caption{\label{fig:clique-to-biclique-transform} Example of a temporal clique \(C\)
    transformed into a temporal bi-clique \(D\).}
\end{figure}
Consider the bi-clique \(D = \biclique{V}{V'}{\edgeLabel'}\) and assume $S_D$ is a bi-spanner in $D$.
  Let \(S_C \coloneqq \edgeLabel^{-1}(\edgeLabel'(S_D) \setminus \{\mu\})\), that is, all
  edges in \(C\) for which an edge with the same time-label is included in
  \(S_D\). As \(\edgeLabel\) is injective, we directly observe that \(|S_C| \leq |S_D|\).

  To show that \(S_C\) is a spanner in \(C\), consider \(u,v\in V\). We need to
  show, that there is a temporal \(u\)\(v\)-path only using edges of \(S_C\). As
  \(S_D\) is a temporal bi-spanner in \(D\), there is a temporal \(uv'\)-path \(P
  \subseteq S_D\).
  Because \(\mu\) is the largest label and each vertex has exactly one adjacent edge with label \(\mu\), if \(P\) contains an edge with
  time-label \(\mu\), it has to end with this edge.
  If that is the case, consider \(P\) without this last edge, which by construction is an \(uv\)-path in \(D\).
  Since $\lambda$ is injective, notice that the path using the same labels in the same order is a temporal path between $u$ and $v$ in $S_C$.
\end{proof}
\fi

This is enough to arrive at the first half of the tight bound we aim to prove in the section.

\iflong
\begin{theorem}
\fi
\ifshort
\begin{theorem}[$\star$]
\fi
\label{cor:spanner-leq-bispanner}
  For \(n\in\N\), we have \(\cliqueSpannerSize{n} \leq \bicliqueBispannerSize{n}\).
\end{theorem}

\iflong
\begin{proof}
  Let \(C = (V ,\edgeLabel)\) be a temporal clique with \(|V| = n\), such that any minimum spanner \(S^{*}\) in \(C\) obeys \(\abs{S^{*}} = \cliqueSpannerSize{n}\).

  By
  \Cref{lem:clique-spanner-by-bispanner}, there is a temporal bi-clique \(D\) with $n$ vertices per side such
  that for any minimum bi-spanner \(S_D\) there is a spanner \(S_C\) in \(C\) with
  \(|S_C| \leq |S_D| \leq \bicliqueBispannerSize{n}\). By definition of \(S^{*}\),
  we conclude \( \cliqueSpannerSize{n} = |S^{*}|\leq |S_C| \le \bicliqueBispannerSize{n}\).
\end{proof}
\fi

We now show that \(\cliqueSpannerSize{2n}\ge\bicliqueBispannerSize{n}\).
To this end, we take a bi-clique and add edges within the parts that yield
no benefit in the construction of a spanner.

\begin{lemma}
\label{lem:spanner-includes-bispanner}
  Let \(D=(A,B, \edgeLabel)\) be a temporal bi-clique with \(\abs{A} = \abs{B}\).
  Then there is a temporal clique \(C = (A\sqcup B, \edgeLabel')\) such that for any spanner \(S\) in \(C\) the set \(S \cap (\bicliqueEdges{A}{B})\) is a bi-spanner in \(D\).
\end{lemma}

\begin{proof}
    We define \(\edgeLabel'\) such that for distinct \(v,w\in A\sqcup B\), we have \[
      \edgeLabel'(\undirEdge{v}{w})= \begin{cases}
    0, &\caseIf v,w\in B;\\
    1+\edgeLabelv{v}{w}, &\caseIf v\in A, w\in B \text{ or } v\in B, w\in A;\\
    \mu                   &\caseIf v,w \in A,
    \end{cases}
  \] where \(\mu>1+ \max_{e\in \bicliqueEdges{A}{B}}\edgeLabele{e}\).
    Intuitively, all early edges are between vertices in \(B\) and all late
    edges are between vertices in \(A\).
    The original edges keep their ordering.

    Let \(S\) be a spanner of \(C\).
    To show \(S \cap (\bicliqueEdges{A}{B})\) is a bi-spanner in \(D\), let  \(a\in A\) and  \(b\in B\).
    As \(S\) is a spanner, there is a temporal \(ab\)-path \(P \subseteq S\).
    If  \(P\) contains any edges within \(A\), no edge to \(B\) can be used
    afterwards.
    Similarly, once the path arrives at \(B\), it is already too late to use any
    edge within \(B\).
    Thus, \(P \subseteq \bicliqueEdges{A}{B}\) and \(S\cap (\bicliqueEdges{A}{B})\) is a bi-spanner in \(D\).
\end{proof}

This yields the following relationship between clique and bi-clique spanner sizes.

\iflong
\begin{theorem}
\fi
\ifshort
\begin{theorem}[$\star$]
\fi%
\label{lem:spanner-geq-bispanner}
For \(n \in \N\), we have
\(\cliqueSpannerSize{2n}\ge\bicliqueBispannerSize{n}\).
\end{theorem}

\iflong
\begin{proof}
  Let \(D=(A,B, \edgeLabel)\) be a temporal bi-clique with \(|A|=|B| = n\),
  such that any minimum bi-spanner \(S^{*}\) in \(D\) obeys \(\abs{S^{*}} = \bicliqueBispannerSize{n}\).
  By \Cref{lem:spanner-includes-bispanner}, there is a temporal clique
  \(C = (A\sqcup B, \edgeLabel')\) such that any spanner \(S\) in \(C\) satisfies
  \(\abs{S} \geq \abs{S^{*}}\). Since the vertex set of $C$ has size $2n$, there is a spanner of size at most \(\cliqueSpannerSize{2n}\) in \(C\), and thus we have \(\cliqueSpannerSize{2n} \geq \bicliqueBispannerSize{n}\).
\end{proof}
\fi

To conclude this section, we just need to show, that
\(\bicliqueBispannerSize{n} \geq \frac{1}{4}\bicliqueBispannerSize{2n}\). To that end,
we show the slightly more general statement that for any \(k\in \N^{+}\), we have
\(k^{2}\bicliqueBispannerSize{n} \geq \bicliqueBispannerSize{kn}\), which we
achieve by splitting the instances into smaller subinstances.

\iflong
\begin{lemma}
\fi
\ifshort
\begin{lemma}[$\star$]
\fi
\label{lem:biclique-weak-recurrence}
  For all \(k, n\in\N^{+}\), we have \(k^{2}\bicliqueBispannerSize{n} \geq \bicliqueBispannerSize{kn}\).
\end{lemma}

\iflong
\begin{proof}
  Consider any temporal bi-clique \(D=\biclique{A}{B}{\edgeLabel}\),
  where \(\abs{A}=\abs{B} = kn\). Consider a partition of \(A\) and \(B\) into $k$ equally sized blocks
  \(A_{1}, A_{2}, \dots, A_{k} \subseteq A\) and \(B_{1}, B_{2},\dots,B_{k}\subseteq B \). For any
  \(i,j \in [k]\) define the temporal graph \(D_{i,j} =D[A_{i}\sqcup B_{j}]\) to be the
  sub-graph of \(D\) induced on the respective blocks with the labeling \(\edgeLabel\)
  restricted  appropriately and let \(S_{i,j}\) be a minimum
  bi-spanner on \(D_{i,j}\). We now set \(S \coloneqq \bigcup_{i,j\in [k]}S_{i,j}\).

For any \(i,j \in [k]\), we have \(\abs{A_{i}}=\abs{B_{j}}=n\) and therefore
\(\abs{S_{i,j}}\leq \bicliqueBispannerSize{n}\) and \(S\) has appropriate size \ie{\(\abs{S}\leq
      k^2\bicliqueBispannerSize{n}\)}.
 Let \(i, j \in [k]\) and \(a \in A_i, b\in B_j\). Then there is a temporal path from $a$ to $b$ using
  only edges of \(S_{i,j}\subseteq S\), so $S$ is a bi-spanner for $D$.
\end{proof}
\fi

\begin{remark}\label{rem:biclique-recurrence}
  We later strengthen this statement using tools developed in \autoref{sec:structure-bi-cliques}. Concretely, considering
  %the subinstances 
  \(D_{i} \coloneqq D[A_{i}\cup B]\), we show
  \(2k(k-1)n + k\bicliqueBispannerSize{n} \geq \bicliqueBispannerSize{kn}\).
\end{remark}

To plug all pieces together, we need a final property about $\bicliqueBispannerSizeOp$.

\iflong
\begin{lemma}
\fi
\ifshort
\begin{lemma}[$\star$]
\fi
\label{lem:bispanner-size-monotone}
  The worst-case size of spanners for bi-cliques is monotonic in the size of the bi-clique, that is for all $n \in \N$, we have $\bicliqueBispannerSize{n} \le \bicliqueBispannerSize{n+1}$. The same holds for $\cliqueSpannerSizeOp$.
\end{lemma}

\iflong
\begin{proof}
  We only prove the statement for bi-cliques and $\bicliqueBispannerSizeOp$, the proof for cliques and $\cliqueSpannerSizeOp$ is analogous.
  Let $n \in \N$ and let $\defaultBiclique$ be a bi-clique of size $n$. We construct a bi-clique of size $n+1$ with at least the same minimum spanner size as $D$ by cloning one vertex on each side. To be exact, take $a \in A$ and $b \in B$ and let
  \begin{align*}
    A' &\coloneqq A \cup \{a'\}, \\
    B' &\coloneqq B \cup \{b'\}, \\
    \edgeLabel' &\coloneqq x, y \mapsto \begin{cases}
      \edgeLabelv{a}{b}, \quad &\text{if } x=a', y=b', \\
      \edgeLabelv{a}{y}, \quad &\text{else if } x=a', \\
      \edgeLabelv{x}{b}, \quad &\text{else if } y=b', \\
      \edgeLabelv{x}{y}, \quad &\text{otherwise},
    \end{cases} \\
    D' &\coloneqq (A' \sqcup B', \edgeLabel').
  \end{align*}

  Now any spanner $S'$ for $D'$ can be transformed into a spanner $S$ of smaller or equal size for $D$, by replacing any edge in $S'$ adjacent to $a'$ or $b'$ with the corresponding edge adjacent to $a$ respective $b$. Any temporal path in $S'$ now yields a temporal path in $S$ where we can use $a$ and $b$ instead of $a'$ and $b'$ since the labels of the original and cloned edges are the same.
\end{proof}
\fi
Now all pieces are in place to puzzle together the proof of our main theorem.

\begin{proof}[Proof of \cref{thm:bi-cliques-are-cliques}]
    Let $n \in \N$ and $c \in \{0,1,2,3\}$ be such that $n\geq 3$ and $n+c$ is a multiple of 4. 
    In the following equation, we use \Cref{lem:bispanner-size-monotone} in the first and last step, apply \cref{lem:biclique-weak-recurrence} in the third step and \cref{lem:spanner-geq-bispanner} in the forth step, to obtain 
  \[
    \bicliqueBispannerSize{n}
    \le \bicliqueBispannerSize{n+ c}
    = \bicliqueBispannerSize{4 \cdot \frac{n+c}{4}}
    \le 16 \cdot \bicliqueBispannerSize{\frac{n+c}{4}}
    \le 16 \cdot \cliqueSpannerSize{\frac{n+c}{2}}
    \le 16 \cdot \cliqueSpannerSize{n}.
  \]

  \cref{cor:spanner-leq-bispanner} directly gives us $\cliqueSpannerSize{n} \le \bicliqueBispannerSize{n}$, allowing us to conclude \(\cliqueSpannerSize{n} \in \Theta(\bicliqueBispannerSize{n})\).
\end{proof}

% Keep this at the end of each file
%%% Local Variables:
%%% mode: latex
%%% TeX-master: "../writeup.tex"
%%% End:

\section{Structure of Bi-Clique Instances}\label{sec:structure-bi-cliques}

Using \Cref{thm:bi-cliques-are-cliques}, from now on we only consider bi-cliques. Our first
structural insight for this type of instance was already used by Casteigts, Peters, and Schoeters~\cite{casteigts_fireworks_conf} on special reduced instances of bi-cliques.
Remember, that we call an instance extremally matched if the set of all earliest edges $\earlyMatching$ and the set of all latest edges $\lateMatching$ each form a perfect matching.
We give a general reduction rule which reduces the graph in case the instance is not extremally matched.
This structural property, that we can assume the existence of such perfect
matchings or reduce the instance, is useful for almost all further constructions.
For example, it enables us to reduce the instance size, when the number of vertices per side is not equal.
Additionally, we can think of the sets \(A\) and \(B\) of the bi-clique as interchangeable. For
example, consider any set of edges that includes temporal paths from some set $A'
\subseteq A$ to all vertices in $A$. We can extend these paths to reach all vertices in $B$ by including
the set of all latest edges, because it is a perfect matching between $A$ and $B$.

We start with proving that
every \(b\in B\) is incident to at most one edge in \(\earlyMatching\) or we
can reduce. The same holds for any \(a \in A\) with respect to
\(\lateMatching\).
For our reduction rules, we adapt the concept of dismountability on temporal
cliques, first defined in~\cite{casteigts_fireworks_conf}, to
temporal bi-cliques.

\begin{definition}[dismountable]\label{def:dismountable}
    Vertex $a \in A$ is \emph{dismountable} if there is $a' \in
    A$ with $a \ne a'$, s.t. $\leqAt{a}{a'}{\earlyPartner{a'}}$. Vertex $b \in B$ is dismountable if 
    there is $b' \in B$ with $b \ne b'$, s.t. $\leqAt{b'}{b}{\latePartner{b'}}$.
\end{definition}

Intuitively, vertex~$a$ can delegate its obligation to reach all vertices in $B$ to vertex~$a'$
in exchange for including two edges in the bi-spanner. Similarly, vertex~$b$ can be reached by
all \nodes{} in $A$ that can reach vertex~$b'$ if we include the two necessary edges in the bi-spanner.

Note that this corresponds to the definition in~\cite{casteigts_fireworks_conf}, but transferred to the case of bi-spanners, where all \nodes{} in $A$ only need to reach \nodes{} in $B$.
Concretely, this means that vertex~$a$ can reach vertex~$a'$ via the earliest edge of
$a'$ on a temporal path of length~2. Therefore, if this configuration occurs in
the bi-clique, we can remove $a$ and include the edges
$\undirEdge{a}{\earlyPartner{a'}}$ and $\undirEdge{a'}{\earlyPartner{a'}}$ into
our bi-spanner. In the final solution for the whole bi-clique, vertex~$a$ then has a
temporal path to every vertex~$b \in B$ via vertex~$a'$. Similarly, if a vertex $b
\in B$ is dismountable, every temporal path that ends in $b'$
can be extended to be a temporal path ending in vertex~$b$, and we can remove
$b$ from the bi-clique.

The following lemma is similar to the one in \cite{casteigts_fireworks_conf}. For completeness, we prove it again.
\iflong
\begin{lemma}
\fi
\ifshort
\begin{lemma}[$\star$]
\fi
\label{lem:reduce-early-late}
    Let $\defaultBiclique$ be a bi-clique. Let $a \in A$, $\earlyPartner{a} = b$, then $a = \earlyPartner{b}$ or $\earlyPartner{b}$ is
    dismountable. Let $b \in B$, $\latePartner{b} = a$, then $\latePartner{a} = b$ or $\latePartner{a}$ is
    dismountable.
\end{lemma}

\iflong
\begin{proof}
    Let $a \in A$ with $b = \earlyPartner{a}$. By definition of $\earlyMatching$, we have $\leqAt{\earlyPartner{b}}{a}{b}$. If $\earlyPartner{b} \ne a$, $\earlyPartner{b}$ is dismountable through $a$.
    The other case follows in the same fashion.
\end{proof}
\fi

We can apply \Cref{lem:reduce-early-late} to show that we can always
find a dismountable vertex whenever the number of vertices per side differs.

\iflong
\begin{lemma}
\fi
\ifshort
\begin{lemma}[$\star$]
\fi
\label{lem:equal-size-sides}
    For $\defaultBiclique$, if $|A| \ne |B|$ then a dismountable vertex exists.
\end{lemma}

\iflong
\begin{proof}
    If $|A| > |B|$, by the pigeonhole principle, there must be at least two
    vertices in $A$ with the same earliest neighbor. Similarly, if $|A| < |B|$,
    there must be at least two vertices in $B$ with the same latest neighbor.
    Due to \Cref{lem:reduce-early-late} there must be a dismountable vertex in these cases. 
\end{proof}
\fi

Since from now on, both parts of $D$ have the same size after removing
dismountable vertices, we write $n \coloneqq |A| = |B|$ and call $n$ the size
of the instance. This justifies also that the definition of $\bicliqueBispannerSizeOp$ assumes equal part sizes and gives the following corollary. 

\begin{corollary}\label{cor:rim-matching-reduction}
    Let $\defaultBiclique$ and $|A| \le |B|$. Then, there is a bi-spanner for $D$ of size $\bicliqueBispannerSize{|A|} + 2(|B| - |A|)$. Similarly, if $|B| \le |A|$ there is a bi-spanner of size $\bicliqueBispannerSize{|B|} + 2(|A| - |B|)$.
\end{corollary}

Finally, combining our lemmas yields that \(\earlyMatching\) and \(\lateMatching\) are extremally matched. 

\iflong
\begin{theorem}[Extremal Matching]
\fi
\ifshort
\begin{theorem}[Extremal Matching),($\star$]
\fi
\label{thm:rim-matching}
    Let $\defaultBiclique$ be a bi-clique. Then $D$ is extremally matched or there is a dismountable vertex. 
\end{theorem}

\iflong
\begin{proof}
    We only prove the statement for $\earlyMatching$, the proof for
    $\lateMatching$ follows analogously. Suppose there is no dismountable
    vertex. Let $b \in B$. By definition, $b$ has at least one incident edge in
    $\earlyMatching$, namely $\undirEdge{b}{\earlyPartner{b}}$. If there is an
    $a \in A$ with $\earlyPartner{a} = b$, due to \Cref{lem:reduce-early-late},
    we have $\earlyPartner{b} = a$, thus $\pi^-$ has at most one edge adjacent
    to $b$. Therefore $b$ has exactly one incident edge and with
    \Cref{lem:equal-size-sides}, $\earlyMatching$ is a perfect matching.
\end{proof}
\fi

After removing dismountable vertices, we now know that the remaining graph must be extremally matched.
Since reducing an instance to be extremally matched only requires the inclusion of linearly many edges, we focus from now on extremally matched bi-cliques. 
This also allows us to switch the roles of \(A\) and \(B\), which will be useful later on.

\iflong
\begin{lemma}
\fi
\ifshort
\begin{lemma}[$\star$]
\fi
\label{lem:switch-sides}	
    Let $\defaultBiclique$ be an extremally matched bi-clique and let $S$ be a bi-spanner for $D$.
    If $\earlyMatching, \lateMatching \subseteq S$, then $S$ is also a bi-spanner for $\biclique{B}{A}{\edgeLabel}$.
\end{lemma}

\iflong
\begin{proof}
	Let $b \in B$ and $a \in A$.
    We look at the path $b \to \earlyPartner{b} \overset{S}{\leadsto} \latePartner{a} \to a$, which is included in $S$ by choice.
    Since $\earlyMatching$ (respectively $\lateMatching$) forms a matching, the edge $\undirEdge{b}{\earlyPartner{b}}$ ($\undirEdge{\latePartner{a}}{a}$) is earliest (latest) for $\earlyPartner{b}$ ($\latePartner{a}$).
    Thus, this path is temporal.
\end{proof}
\fi

With \Cref{cor:rim-matching-reduction} we prove
\Cref{lem:biclique-recurrence}. This gives a simpler proof that
$\cliqueSpannerSize{n} \in \bigO(n \log n)$.

\iflong
\begin{lemma}
\fi
\ifshort
\begin{lemma}[$\star$]
\fi
\label{lem:biclique-recurrence}
    For all $k, n \in \N^+$, we have $2k(k-1)n + k\bicliqueBispannerSize{n} \ge \bicliqueBispannerSize{kn}$.
\end{lemma}

\iflong
\begin{proof}
    The proof is similar to the proof of
    \Cref{lem:biclique-weak-recurrence}. First, consider any temporal bi-clique
    $\defaultBiclique$, where $|A| = |B| = kn$. We partition $A$
    into equally sized blocks $A_1, A_2, \ldots, A_k \subseteq A$. For $i \in
    [k]$, let $D_i \coloneqq D[A_i \sqcup B]$ with appropriately
    restricted labeling.

    For every $D_i$ consider a minimum spanner $S_i$, which satisfies $|S_i|
    \le \bicliqueBispannerSize{n} + 2(k-1)n$ by
    \Cref{cor:rim-matching-reduction}. Set $S \coloneqq \bigcup_{i \in [k]}
    S_i$. Then, $S$ has the correct size and for every $i \in [k], a \in A_i$,
    and $ b \in B$, there is a temporal $ab$-path using only edges of
    $S_i\subseteq S$.
\end{proof}
\fi

\iflong
\begin{figure}[t]
    \SetAlgoNoLine%
    \centering
    \begin{tcolorbox}[blanker,width=(\linewidth-3.98cm)]
    \begin{algorithm}[H]
        \DontPrintSemicolon
        \SetKwFunction{FMain}{biSpanner}
        \SetKwFunction{dis}{dismount}
        \SetKwProg{Fn}{def}{:}{}
        \Fn{\FMain{$\defaultBiclique$}}{
            \lIf{$|A| = 1$}{
                \Return $E(D)$
            }
            partition $A$ into non-empty $A_1, A_2$ with $|A_1| = \lfloor |A| / 2 \rfloor$\;
            $D_1 \gets D[A_1 \sqcup B]$, $D_2 \gets D[A_2 \sqcup B]$\;
            $(D_1^*, S_1^*) \gets$ \dis{$D_1$}, $(D_2^*, S_2^*) \gets$ \dis{$D_2$}\;
            \Return \FMain{$D_1^*$} $\cup$ \FMain{$D_2^*$} $\cup$ $S_1^*$ $\cup$ $S_2^*$
        }
    \end{algorithm}
    \end{tcolorbox}
    \caption{Our algorithm for computing $\bigO(n\log n)$ bi-spanners based on \Cref{lem:biclique-recurrence}. The function \texttt{dismount} takes a bi-clique, exhaustively dismounts all dismountable vertices, and returns the remaining instance and the included edges. By \Cref{lem:equal-size-sides}, this dismounts at least until $|A| = |B|$.}\label{alg:nlogn}
\end{figure}
\fi

If we now choose $k>1$, 
we get
 $\cliqueSpannerSize{n} \in \bigO(n \log n)$ using the
master theorem \cite{cormen2022introduction}. Note that we can also implement
all the necessary reductions to find such a spanner in linear time in the number
of edges in the temporal clique. This achieves the same spanner size as in~\cite{casteigts_fireworks_conf} but with a simpler algorithm and proof. 
\iflong
See \Cref{alg:nlogn} for an outline of the algorithm.
\fi
\ifshort
See the appendix for an outline of the algorithm.
\fi

\begin{corollary}
    We have $\cliqueSpannerSize{n} \in \bigO(n \log n)$ and there is an algorithm
    that computes a spanner of size $\bigO(n \log n)$ in time $\bigO(n^2)$.
\end{corollary}

% Keep this at the end of the file
%%% Local Variables:
%%% mode: latex
%%% TeX-master: "../writeup.tex"
%%% End:

\section{Pivotable Bi-Cliques}\label{sec:pivotable}
An important structure that was used extensively in \cite{casteigts_fireworks_conf,casteigts_threshold} is the
notion of a pivot. A pivot is a vertex that can be
reached by all vertices until a certain time point $t$ and that can also
reach all other vertices starting at time point $t$.
In~\cite{casteigts_threshold} it is shown that this structure, though intuitively quite restrictive, appears with high
probability in random temporal graphs.

However, there are instances without this structure~\cite{casteigts_fireworks_conf}. Thus, we study an even
more widely applicable structure, based on a what we call a \emph{$c$-pivot-edge} or \emph{partial pivot-edge} (when omitting the concrete parameter). We exploit partial pivot-edges to apply the divide-and-conquer
paradigm to find linear-size temporal spanners. To the best of our knowledge, this paradigm has not
yet been applied to this problem.
As we see later, the absence of a partial pivot-edge yields strong structural properties and it heavily restricts the class of temporal cliques where no linear spanner can be found using current techniques.

\iflong
\subsection{Partial Pivot-Edges}%
\fi

Recall the definition of $\inSetOp$ and $\outSetOp$ in \Cref{sec:prelims}. 
We observe that in extremally matched bi-cliques, 
the sets \(\inSet{e}\) and \(\outSet{e}\) are always distributed
evenly between \(A\) and \(B\).

\iflong
\begin{lemma}
\fi
\ifshort
\begin{lemma}[$\star$]
\fi
\label{lem:inset-outset-symm}
	Let \(G=(A,B, \lambda)\) be an extremally matched bi-clique. For all
\(e \in \bicliqueEdges{A}{B}\), we have \(\abs{\inSet{e}\cap A} = \abs{\inSet{e}\cap B}\) and
\(\abs{\outSet{e}\cap A} = \abs{\outSet{e}\cap B}\).
\end{lemma}

\iflong
\begin{proof}
  Consider any \(e \in A\otimes B\). Notice that for all \(v\in \inSet{e}\), we have
  \(\earlyPartner{v} \in \inSet{e}\) and \(\earlyMatching\) is a bijection between
  \(A\) and \(B\). So, there is a bijection between \(\inSet{e}\cap A\) and
  \(\inSet{e}\cap B\). Thus, \(\abs{\inSet{e}\cap A} = \abs{\inSet{e}\cap B}\). The claim
  regarding \(\outSet{e}\) follows analogously.
\end{proof}
\fi

Next, we show that we can reach all vertices in $\outSet{e}$ from $e$ starting at or after $\edgeLabele{e}$, using exactly $\left| \outSet{e} \right| - 1$ edges.
The analogous statement holds for the vertices in $\inSet{e}$, in the sense that those can
reach $e$ at or after $\edgeLabele{e}$, using only
$\left| \inSet{e} \right| -1$ many edges.

\begin{lemma}\label{lem:rooted_tree}
	Let $e = \undirEdge{a}{b} \in \bicliqueEdges{A}{B}$. All vertices in $\outSet{e}$ can be organized in a tree~$T$ rooted at $a$, s.t. $T$ only has edges with label at least $\edgeLabele{e}$ and for each $v \in \outSet{e}$ there are temporal paths $a \leadsto v$ and $b \leadsto v$ in $T$.
	The analogous result holds for vertices in $\inSet{e}$. 
\end{lemma}

The proof of \Cref{lem:rooted_tree} is already known \cite{xuan2003computing} and it relies on constructing a foremost tree for $\outSet{e}$ (and its inverted version for $\inSet{e})$. It was considered in \cite{lavrov2016increasing} and lately has been used in~\cite{casteigts_threshold}, where foremost trees are applied to temporal cliques.  With this idea, we can temporally connect all \(u \in \inSet{e}\) to all
\(v \in \outSet{e}\) using a linear number of edges.

\iflong
\begin{corollary}
\fi
\ifshort
\begin{corollary}[$\star$]
\fi\label{cor:pivot-connections}
	Given an edge $e \in \bicliqueEdges{A}{B}$, we can connect the vertices in $\inSet{e}$ to those in $\outSet{e}$ by using at most $\left| \inSet{e} \right| + \left|\outSet{e}  \right| - 3$ many edges.
\end{corollary}

\iflong
\begin{proof}
	Since we can organize the vertices $\inSet{e}$, $\outSet{e}$ as a tree by \Cref{lem:rooted_tree}, all
	vertices in $\inSet{e}$ can reach $e$ and, afterwards, can reach
	$\outSet{e}$ using the respective trees. Both trees consist of
	$\left| \inSet{e} \right|-1, \left| \outSet{e} \right| -1$ many edges
	respectively. But since \(a\) and \(b\) appear in both sets being connected
	by $e$, we double count this edge 
	resulting in $\left| \inSet{e} \right| + \left|\outSet{e}  \right| - 3$ many
	edges.
\end{proof}
\fi

To apply divide-and-conquer, we choose an edge $e \in \bicliqueEdges{A}{B}$ and include the edges according to \Cref{cor:pivot-connections}. Then, we create two sub-instances to connect the remaining vertices. 

\iflong
\begin{theorem}
\fi
\ifshort
\begin{theorem}[$\star$]
\fi
\label{thm:edge-pivot}
	Let $\defaultBiclique$ be an extremally matched bi-clique with
$n \coloneqq \left| A \right| = \left| B \right|$. There exists a bi-spanner in
\(D\) with size at most
\[\min_{e \in \bicliqueEdges{A}{B}}\left( \bicliqueBispannerSize{n- \frac{\left|\inSet{e} \right|}{2}} + \bicliqueBispannerSize{n- \frac{\left|\outSet{e} \right|}{2}} + 2\abs{\inSet{e}} + 2 \abs{\outSet{e}} - 3 \right).\]
\end{theorem}

\iflong
\begin{proof}
	Consider any \(e\in\bicliqueEdges{A}{B}\).
	By \Cref{cor:pivot-connections}, there is a set of edges \(F_{e}\) such that
	\(|F_{e}| = \abs{\inSet{e}} + \abs{\outSet{e}} - 3\) and all vertices in
	\(\inSet{e}\) can reach all vertices of \(\outSet{e}\) using the edges in~\(F_{e}\).
	To connect the remaining \nodes{}, we create two sub-instances
	$G'_{e}\coloneqq G[(A \setminus \inSet{e}) \sqcup B]$ and
	$G''_{e} \coloneqq G[A \sqcup (B \setminus \outSet{e})]$. 
	Let $S'_e$ be a minimum bi-spanner of $G'_e$, and $S''_e$ be a minimum bi-spanner of $G''_e$. 
	By \Cref{lem:inset-outset-symm}, the smaller sides of $G'_e$ and $ G''_e$ have
size $n - \frac{\left| \inSet{e} \right|}{2}$ and
$ n - \frac{\left| \outSet{e} \right|}{2}$, respectively.
	By \Cref{cor:rim-matching-reduction}, we get that 
	$|S'_{e}| \leq \bicliqueBispannerSize{n- \frac{\left| \inSet{e} \right|}{2}} +	\abs{\inSet{e}}$
	and
	$|S''_{e}| \leq \bicliqueBispannerSize{n- \frac{\left| \outSet{e} \right|}{2}} + \abs{\outSet{e}}$.

	We claim that the set \(S_{e}\coloneqq F_{e}\cup S'_{e}\cup S''_{e}\) is a bi-spanner for \(G\). Consider
	any vertex~\(a \in A\) and vertex~\(b \in B\). If vertex~\(a\notin \inSet{e}\), then vertex~\(a\) can reach vertex~\(b\) via
	\(S'_{e}\). If \(b\notin\outSet{e}\), vertex~\(a\) can reach vertex~\(b\) via \(S''_{e}\).
	Otherwise, \(a\in\inSet{e}\) and \(b\in\outSet{e}\) and so \(a\) can reach \(b\)
	via \(F_{e}\). Choosing the edge where the set \(S_{e}\) is smallest, yields
	the result.
\end{proof}
\fi

As \(2\abs{\inSet{e}} + 2\abs{\outSet{e}} -3\leq 8n - 3\), we see that only a linear
number of edges is included to split the instance.  With this, we can show that if for some edge
\(e \in \bicliqueEdges{A}{B}\) the sets \(\inSet{e}\) and \(\outSet{e}\) overlap, we can reduce the instance significantly. We name the set ${\inSet{e} \cap \outSet{e}}$ the \emph{pivot-set} of $e$.
If $\abs{\inSet{e} \cap \outSet{e}} \geq 2cn$ for $c \in (0, 1]$, we call $e$ a \emph{$c$-pivot-edge}. When omitting the concrete constant $c$, we refer to $e$ as a \emph{partial pivot-edge}.

\iflong
\begin{theorem}
\fi
\ifshort
\begin{theorem}[$c$-Pivot-Edge),($\star$]
\fi
\label{cor:pivot-edge-spanner}
	Let $c \in (0, 1]$. Let \(D = (A,B,\lambda)\) be an extremally matched bi-clique with
	\(n\coloneqq |A| = |B|\). If there is a $c$-pivot-edge in $D$, we can split the instance into two sub-instances with overall number of vertices reduced by $2cn$.
    Per removed vertex, we include at most \(\frac{4}{c}\) edges.
	If applicable recursively, we get a spanner of size \(\frac{8}{c}n\).
\end{theorem}

\iflong
\begin{proof}
    Let \(e \in\bicliqueEdges{A}{B}\) be such that \(\abs{\inSet{e} \cap \outSet{e}} \geq 2cn\).
	Note that
	$\left| \inSet{e} \right| + \left| \outSet{e} \right| = \left| \inSet{e} \cup \outSet{e} \right| + \left| \inSet{e} \cap \outSet{e} \right| \ge \left( 2 + 2c \right) n$.
	By \Cref{thm:edge-pivot}, the sub-instances have sides of overall size at most $2n - \frac{2+2c}{2}n = \left( 1-c \right) n$.
	This means that $2cn$ vertices have been removed using at most \(8n -3\) edges
	overall. Per such vertex, we include
	\(\frac{8n - 3}{2cn} \leq \frac{4}{c}\) edges.

	To recursively apply \Cref{thm:edge-pivot}, observe that we can reduce any dismountable
	vertex until both sub-instances are extremally matched. For each vertex
	removed in this way, we include \(2\leq \frac{4}{c}\) edges.
	As the number of edges included per removed vertex in both ways does not
	depend on \(n\), we include at most \(\frac{4}{c}\abs{A\sqcup B}= \frac{8}{c}n\) edges.
\end{proof}
\fi

This theorem tells us that if we repeatedly find such a $c$-pivot-edge
$e\in \bicliqueEdges{A}{B}$ with $\left| \inSet{e} \cap \outSet{e} \right|$ being at
least a (previously fixed) $c$-fraction of the current size of the instance, we
get a linear-size bi-spanner for our graph. 
Therefore, if we search for counterexamples for
linear-size spanners in temporal bi-cliques, we can assume that for any fixed $c$ no
such $c$-pivot-edge is present, that is, for all $e \in \bicliqueEdges{A}{B}$, we have
$\left| \inSet{e} \cap \outSet{e} \right| \in \smallO(n)$.

\ifshort
This has many interesting consequences, removing large classes of graphs from consideration. To illustrate, let us quickly mention three forbidden structures.
First, for \(c\in (0,1]\), no edge \(e\) with \(\abs{\edgeIndex{v}{e}-\edgeIndex{w}{e}}\geq cn\) can exist. We call $e$ \emph{$c$-steep} because, intuitively, it makes a steep jump between the ordering at its two endpoints.
Second, in a bi-clique $\defaultBiclique$ without $c$-steep edges, for any $i, j \in [n]$ and $a \in A$, we know that for the label spread $S \coloneqq \left\{ b \in B \mid i \leq \edgeIndex{b}{a} \leq j \right\}$, we have \( \left| S \right| < j - i + 2cn \). 
Intuitively, the number of vertices for which vertex~$a$ is one of their $i$-th to $j$-th neighbors is relatively small. Of course, the same property holds when $A$ and $B$ are swapped.
Third, we can now assume that all but sublinearly many vertices $v \in \inSet{e}$ and all but sublinearly many $w \in \outSet{e}$ fulfill the condition that $\latePartner{v} \in \inSet{e}$ and $\earlyPartner{w} \in \outSet{e}$. Intuitively, the matchings $\pi^-$ and $\pi^+$ don't cross from $\inSet{e}$ to $\outSet{e}$ or vice-versa.

These three consequences show that we have gained strong structural insights. They allow us and future work to focus on a more concrete set of bi-cliques.
For a more detailed discussion of these consequences, please see the full version.
\fi
\iflong
The following subsection gives an intuition for how restrictive this condition actually is. 
\fi

\iflong
\subsection{Graph Properties leading to Partial Pivot-Edges}%
\label{sub:steep_edges}
Partial pivot-edges generalize pivotable \nodes{} and yield a number of new, interesting structural properties in residual instances. In this section, we examine structures that, if present, imply that there is at least one partial pivot-edge.

We first look at one particular case of partial pivot-edges, which we call \emph{steep
edges}.
Intuitively, an edge is steep, if it belongs to the later edges
for one of the incident vertices and to the earlier edges for the other one.

\begin{definition}[$c$-steep]\label{def:steep}
	Let $\defaultBiclique$ be a bi-clique. For a constant \(c\in
	(0,1]\), we call \(e\) a \emph{\(c\)-steep edge}, if
	\(\abs{\edgeIndex{v}{e}-\edgeIndex{w}{e}}\geq cn\).
\end{definition}

We now show that we can assume edges not to be steep even in artificial
instances, by proving that as long as there is a steep edge, we can apply
\Cref{cor:pivot-edge-spanner}.

\iflong
\begin{lemma}
\fi
\ifshort
\begin{lemma}[$\star$]
\fi
\label{lem:steep-edges-are-pivots}
	Let \(c\in(0,1]\). If there is a \(c\)-steep edge \(e\in
	\bicliqueEdges{A}{B}\), then \(\abs{\inSet{e}\cap\outSet{e}}\geq 2cn\).
\end{lemma}

\iflong
\begin{proof}
	Let \(e \coloneqq \set{a,b}\). W.l.o.g., assume
	\(\edgeIndex{a}{e}\ge\edgeIndex{b}{e}\). We define \(A^-\coloneqq \set{a'\in
	A}{\leqAt{a'}{a}{b}}\) and \(B^+\coloneqq \set{b'\in
	B}{\geqAt{b'}{b}{a}}\). Then, \(A^-\subseteq (\inSet{e}\cap A)\) and \(B^+
	\subseteq (\outSet{e}\cap B)\) holds. By \Cref{lem:inset-outset-symm} and
	\Cref{lem:in-cup-out-eq-all-vertices}, we have 
	\begin{align*}
		\abs{A^-}+\abs{B^+}&\leq \frac{1}{2}(\abs{\inSet{e}}+\abs{\outSet{e}})\\
						   &=\frac{1}{2} (\abs{\inSet{e} \cup \outSet{e}} + \abs{\inSet{e} \cap \outSet{e}})\\
						   &=n+\frac{1}{2}(\abs{\inSet{e}\cap \outSet{e}}).
	\end{align*}
	We now bound $\abs{A^{-}} + \abs{B^{+}}$ from below. We have
	\(\abs{A^-}+\abs{B^+}=\edgeIndex{a}{e}+n-(\edgeIndex{b}{e}-1)\). Using that
	\(e\) is \(c\)-steep, we thus have \(\abs{A^-}+\abs{B^+}\geq n+cn\).
	Combining both inequalities, we get \(\abs{\inSet{e}\cap\outSet{e}}\geq
	2cn\).
\end{proof}
\fi

While steep edges demonstrate how an intuitively common structure cannot be there if assume there are no partial pivot-edges. However, in randomly generated instances we expect edges not to be steep. 
Assume we add edges randomly to the graph in the order of their labels. Then, when adding the
edge \( \set{v,w} \), both \(v\) and \(w\) should already have a similar number
of edges. In fact, one can prove that with high probability, there is no \(c\)-steep in the graph. 

Yet, as every $c$-steep edge is a $c$-pivot-edge, we want to analyze the graphs without such edges in greater detail. We first look at more simple properties showing that any node only has a limited number of neighbors for which it is at least the $i$-th and at most the $j$-th neighbor.

\begin{lemma} [Label Spread]
	\label{lem:spread}
	Let $c \in (0, 1]$. In a bi-clique $\defaultBiclique$ without $c$-steep edges, for any $i, j \in [n]$ and $v \in V$ we know that for the label spread $S \coloneqq \left\{ v' \in N(v) \mid i \leq \edgeIndex{v'}{v} \leq j \right\}$ it holds that \[
		\left| S \right| < j - i + 2cn. 
	\]  
\end{lemma}

\iflong
\begin{proof}
	Let $S \coloneqq \left\{ v' \in N(v) \mid i \leq \edgeIndex{v'}{v} \leq j \right\}$ and choose a $v' \in S$. As the edge $\undirEdge{v}{v'}$ is not $c$-steep, we have that $\edgeIndex{v}{v'} \in (i - cn, j + cn)$. Since $\orderingAt{v}$ is an injection, we have that $|S| \leq \left| (i - cn, j + cn) \right| \leq j+cn - (i - cn) = j - i + 2cn$.
\end{proof}
\fi

The next lemma tells us that if a vertex $v$ can reach an edge $e = \undirEdge{a}{b}$, $v$ can also reach all edges $e'$ incident to either $a$ or $b$ with $\edgeLabele{e'} > \edgeLabele{e}$.

\iflong
\begin{lemma}
\fi
\ifshort
\begin{lemma}[$\star$]
\fi
\label{lem:activity-width}
	Given a temporal bi-clique $\defaultBiclique$. Fix any $b \in B$ and order all $a_1, \ldots, a_n \in A$ according to $\orderingAt{b}$ \ie{$\leqAt{a_i}{a_{i+1}}{b}$ for all $i \in [n-1]$}. Then \[
		\inSet{\undirEdge{a_1}{b}} \subseteq \cdots \subseteq \inSet{\undirEdge{a_n}{b}}
	\] and \[
		\outSet{\undirEdge{a_1}{b}} \supseteq \cdots \supseteq \outSet{\undirEdge{a_n}{b}}.
	\] 
	The same holds when switching roles for $A$ and $B$.	
\end{lemma}

\iflong
\begin{proof}
	Consider $v \in \inSet{\undirEdge{a_i}{b}}$ and let $j > i$. 
Then there exists a temporal path $P$ starting at $v$ and ending in $\undirEdge{a_i}{b}$. Assume w.l.o.g. that the path ends in $b$, else remove the latest edge from the path. By adding the edge $\undirEdge{b}{a_j}$ to the path, the path stays temporal since $\leqAt{a_i}{a_j}{b}$ and $v \in \inSet{\undirEdge{a_j}{b}}$.
	The proof for $\outSetOp$ as well as for switched roles of $A$ and $B$ follows analogously.
\end{proof}
\fi
We can rephrase the last lemma to characterize which vertices remain in $\inSet{\undirEdge{a}{b}} \cap \outSet{\undirEdge{a}{b}}$ when replacing $\undirEdge{a}{b}$ with $\undirEdge{a'}{b}$. If $\undirEdge{a'}{b}$ has a similar label to $\undirEdge{a}{b}$, the sets $\inSetOp$ and $\outSetOp$ will remain roughly the same.

\begin{corollary}
[Activity Width]\label{cor:activity-width}
	Given a temporal bi-clique $\defaultBiclique$. Fix any $b \in B$ and order all $a_1, \ldots, a_n \in A$ according to $\orderingAt{b}$ \ie{$\leqAt{a_i}{a_{i+1}}{b}$ for all $i \in [n-1]$}.

For every $i \in [n]$ we define $M^{A}_{b,i} \coloneqq \inSet{\undirEdge{a_i}{b}} \cap \outSet{\undirEdge{a_i}{b}} \cap A$. Then for every $a \in A$ there exists an interval $[k, \ell]$ such that $a \in M^{A}_{b, i}$ if and only if $i \in [k, \ell]$. The same holds for $M^{B}_{b, i} \coloneqq \inSet{\undirEdge{a_i}{b}} \cap \outSet{\undirEdge{a_i}{b}} \cap B$.
\end{corollary}

Further assuming that no $c$-pivot-edge exists in $G$, for all $b \in B, i \in [n]$ we have $\left| M^A_{b, i} \right| < cn$, thus on average for each $a \in A$, the interval $[k, \ell]$ such that $a \in M^A_{b, i} \iff k \leq i \leq \ell$ has size less than $cn$. 

\fi

%%%%%%%%%%%%%%%%%%%%%%%%%%%%%%%%%%%%%%%%%%%%%%%%%%%%%%%%%%
% Keep this at the end of each file
%%% Local Variables:
%%% mode: latex
%%% TeX-master: "../writeup.tex"
%%% End:

\section{Bi-Cliques With Reverted Edges}\label{sec:non-pivotable}

We introduced partial pivot-edges and observed some
interesting structural properties in bi-cliques which contain no suitable
$c$-pivot-edge. The first hope would be that every bi-cliques admits such a
$c$-pivot-edge. Unfortunately, this is not the case and
we provide a graph class, the \emph{shifted matching
graph}, in which for every edge, the pivot-set is as small as possible
\ie{$\inSet{\undirEdge{a}{b}} \cap \outSet{\undirEdge{a}{b}} = \{a, b\}$}
and none of our reduction rules from \cref{sec:structure-bi-cliques} apply.

Still, we present a novel technique, called \emph{$e$-reverted edges}, to construct small spanners for shifted matching graphs which also provides solutions for a significantly larger class of bi-cliques.
We relate the technique to partial pivot-edges from \cref{sec:pivotable} and argue why $e$-reverted edges solve a class distinct from partial pivot-edges, beyond the concrete example of shifted matching graphs. 
Unfortunately, also this is not applicable to all graph classes, which we later demonstrate with a class that we call \emph{product graphs} in \Cref{sec:swarms}.

We start by giving the definition of the shifted matching graph.

\begin{definition}[Shifted Matching Graph]
    Let $n \in \N^+$, $A \coloneqq \{a_0, \ldots, a_{n-1}\}$, and $B \coloneqq
    \{b_0, \ldots, b_{n-1}\}$. We define the \emph{shifted matching graph} on
    $n$ vertices per side as $\ringShift{n} \coloneqq
    \biclique{A}{B}{\edgeLabel}$, where for every $i,j \in \setExc{n}$ we label
    the edge
    \(
        \edgeLabelv{a_i}{b_j} \coloneqq j - i \xmod n.
    \)
\end{definition}
\iflong
\begin{figure}[t]
  \captionsetup[subfigure]{justification=centering}
  \centering
  \begin{subfigure}{0.49\textwidth}
      \begin{tabular}{lrrrr} 
          \toprule
          \multirow{2}{*}[-2pt]{Vertex}&\multicolumn{4}{c}{Edge Label}\\ \cmidrule{2-5}
          & 0 & 1 & 2 & 3\\
          \midrule 
          $a_0$ & $b_0$ & $b_1$ & $b_2$ & $b_3$ \\
          $a_1$ & $b_1$ & $b_2$ & $b_3$ & $b_0$ \\
          $a_2$ & $b_2$ & $b_3$ & $b_0$ & $b_1$ \\
          $a_3$ & $b_3$ & $b_0$ & $b_1$ & $b_2$ \\
          \bottomrule
      \end{tabular}
      \centering
      \caption{\label{tab:ringshift-a}View from $A$}
  \end{subfigure}
  \begin{subfigure}{0.49\textwidth}
      \begin{tabular}{lrrrr} 
          \toprule
          \multirow{2}{*}[-2pt]{Vertex}&\multicolumn{4}{c}{Edge Label}\\ \cmidrule{2-5}
          & 0 & 1 & 2 & 3\\
          \midrule 
          $b_0$ & $a_0$ & $a_3$ & $a_2$ & $a_1$ \\
          $b_1$ & $a_1$ & $a_0$ & $a_3$ & $a_2$ \\
          $b_2$ & $a_2$ & $a_1$ & $a_0$ & $a_3$ \\
          $b_3$ & $a_3$ & $a_2$ & $a_1$ & $a_0$ \\
          \bottomrule
      \end{tabular}
      \centering
      \caption{\label{tab:ringshift-b}View from $B$}
  \end{subfigure}
  \caption{\label{tab:ringshift-example} $\ringShift{4}$ with vertices $v \in A \sqcup B$ and all their neighbors, ordered by edge label. }
\end{figure}

\Cref{tab:ringshift-example} shows the shifted matching graph for $n = 4$
 vertices on both sides.
\fi
In shifted matching graphs, the earliest and latest neighbors form a
matching, thus we cannot apply the reduction from \cref{thm:rim-matching}.
Note that for all $i \in \setExc{n}$, we have $\pi^-(a_i) = b_i$ and $\pi^+(a_i) =
b_{i-1\xmod n}$. 
Additionally, we observe that each edge at time $t$ has the form $\undirEdge{a_i}{b_{i+t \xmod n}}$. Thus, the edges with label $t$ form a
perfect matching and \(a_i\) sees the $b_j$ in ascending order of indices,
circularly shifted by \(i\) places compared to $a_0$.  

\begin{observation}\label{lem:ringshift-matching}
    Let $t \in \{0, \dots n-1\}$ be a time-point. Then in $\ringShift{n}$, the
    edges with label $t$ form a perfect matching. 
\end{observation}

We claim that this graph has no partial pivot-edges that are suitable for
reduction in the sense of \Cref{cor:pivot-edge-spanner}. To prove that this
method fails on the shifted matching graph, we first characterize the set of
vertices which can be reached from a vertex $a_i$ before and after a given time.
Finally, we prove that the intersection of $\inSetOp$ and $\outSetOp$ for
all edges at $a_i$ is trivial.

\iflong
\begin{lemma}
\fi
\ifshort
\begin{lemma}[$\star$]
\fi
\label{lem:ringshift-in-out}
    Let $i, t \in \setExc{n}$ and $a_i$ and $b_{i+t \xmod n}$ be vertices in
    $\ringShift{n}$. Then using edges with label at most $t$, precisely the
    vertices $a_i, \dots, a_{i+t\xmod n}$ and $b_i, \dots, b_{i+t \xmod n}$ can
    reach $a_i$ and $b_{i+t \xmod n}$.
    Using only edges with label at least $t$,
    $a_i$ and $b_{i+t \xmod n}$ can precisely reach the vertices $b_{i+t \xmod n},
    \dots, b_{i-1 \xmod n}$ and $a_{i+t+1 \xmod n}, \ldots, a_i$. 
\end{lemma}

\iflong
\begin{proof}
    We prove the first part of the statement via induction over $t$. The second part follows analogously by induction over $n-t-1$. For $t = 0$, clearly $a_i$
    as well as $b_i$ can only be reached by themselves and each other. By
    \cref{lem:ringshift-matching}, we can never compose two edges with label
    $0$.

    Now assume the statement holds for all vertices for some $t < n-1$. Thus,
    precisely $a_i, \dots, a_{i+t\xmod n}$ and $b_i, \dots, b_{i+t \xmod n}$
    can reach $a_i$ and precisely $a_{i+1\xmod n}, \dots, a_{i+t+1\xmod n}$ and
    $b_{i+1\xmod n}, \dots, b_{i+t+1 \xmod n}$ can reach $b_{i+t+1\xmod n}$.
    Now we look at the situation, when we allow edges with label $t+1$ as well.
    Since $a_i$ and $b_{i+t+1 \xmod n}$ are connected with an edge with label
    $t+1$, we can take the union of both sets. No additional vertices can now reach $a_i$ and $b_{i+t+1 \xmod n}$, since after taking an edge with label $t+1$ from a different vertex, we are stuck due to 
    \Cref{lem:ringshift-matching}.
\end{proof}
\fi

Using \Cref{lem:ringshift-in-out}, we can see that using
partial pivot-edges in $\ringShift{n}$ leads to no reduction.

\iflong
\begin{lemma}
\fi
\ifshort
\begin{lemma}[$\star$]
\fi
\label{lem:ringshift-no-pivot-edge} For vertices $a_i,b_j$ in $\ringShift{n}$, we have $\inSet{\undirEdge{a_i}{b_j}} \cap
    \outSet{\undirEdge{a_i}{b_j}} = \{a_i, b_j\}$.
\end{lemma}

\iflong
\begin{proof}
    Let $t \coloneqq \edgeLabelv{a_i}{b_j} = j - i \xmod n$ and therefore $b_j =
    b_{i+t\xmod n}$. By \Cref{lem:ringshift-in-out},
    \[\inSet{\undirEdge{a_i}{b_j}} = \{a_i, \ldots, a_{i+t\xmod n}, b_{i}, \ldots, b_{i+t\xmod n}\}\] and
    \[\outSet{\undirEdge{a_i}{b_j}} = \{a_{i+t+1 \xmod n}, \ldots, a_i,
    b_{i+t\xmod n}, \ldots, b_{i -1 \xmod n}\}.\] Looking at the intersection,
    we notice $\inSet{\undirEdge{a_i}{b_j}} \cap \outSet{\undirEdge{a_i}{b_j}}
    = \{a_i, b_j\}$.
\end{proof}
\fi

However, with a new technique we can still construct linear-size spanners
for $\ringShift{n}$. For this, we need the concept of \emph{$e$-reverted edges}.
This construction sheds light on different ideas, such as making use of $\earlyMatching$ and
$\lateMatching$ to switch $A$ and $B$, that can be used for other graphs than
just $\ringShift{n}$.
We start by giving the general technique and then proving that it is applicable to construct a linear-size bi-spanner for the shifted matching graph.

\begin{definition}[$e$-reverted]\label{def:reverted}
For a bi-clique $\defaultBiclique$, consider an edge $e = \undirEdge{a}{b}$. We say that an edge $\undirEdge{a'}{b'}$ is \emph{$e$-reverted}, if $\leqAt{a'}{\latePartner{b'}}{b}$ or $\leqAt{\earlyPartner{a'}}{b'}{a}$. Denote the set of edges that are not $e$-reverted as $\malicious{e} \coloneqq \set{\undirEdge{a'}{b'} \in \bicliqueEdges{A}{B}}{\undirEdge{a'}{b'} \text{ is not $e$-reverted}}$.
\end{definition}

Intuitively, for each edge $e$ there is a linear-size set of edges that connects all pairs of vertices $\undirEdge{a'}{b'}$ that form an $e$-reverted edge.
If $\leqAt{a'}{\latePartner{b'}}{b}$, the path $a' \to b \to \latePartner{b'} \to b'$ is temporal, since $b'$ is also the latest neighbor for $\latePartner{b'}$.
If $\leqAt{\earlyPartner{a'}}{b'}{a}$, the path $a' \to \earlyPartner{a'} \to a \to b'$ is temporal, since $a'$ is also the earliest neighbor $\earlyPartner{a'}$.
See~\Cref{fig:rev}.
\begin{figure}[t]
    \centering
    \includegraphics[width=42.42mm]{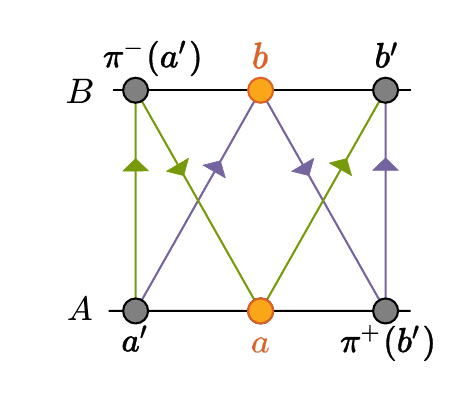}
    \caption{If at least one of the green or purple paths is temporal, $\{a', b'\}$ is $\undirEdge{a}{b}$-reverted.}
    \label{fig:rev}
\end{figure}
Notice that, no matter which $a' \in A$ and $b' \in B$ we choose, the paths only use edges incident to $a$ or $b$ and edges from $\earlyMatching$ or $\lateMatching$. We use this to construct a small bi-spanner for $D$, if $\malicious{e}$ is small.

\iflong
\begin{theorem}
\fi
\ifshort
\begin{theorem}[$\star$]
\fi
\label{thm:malicious-solution}
    Let $\defaultBiclique$ be an extremally matched bi-clique of size $n$ and $e = \undirEdge{a}{b} \in \bicliqueEdges{A}{B}$. Then $D$ has a bi-spanner of size at most $4n - 4 + \abs{\malicious{e}}$.
\end{theorem}

\iflong
\begin{proof}
    Define $S$ to include all edges adjacent to $a$ or $b$ as well as the
    perfect matchings $\earlyMatching$ and $\lateMatching$ and all edges in $\malicious{e}$.
    First notice that $|S| \le 4n + \abs{\malicious{e}}$ and we double count at least the first and last edges of $a$ and $b$ (and $e$ if these overlap).

    Furthermore, we show that $S$ is a bi-spanner for $D$. Let $a' \in A$ and $b' \in B$. We distinguish three cases.

    Case 1: $\undirEdge{a'}{b'} \in \malicious{e}$. In this case, the direct edge between $a'$ and $b'$ is contained in $S$. In all other cases $\undirEdge{a'}{b'}$ is $e$-reverted.

    Case 2: $\leqAt{a'}{\latePartner{b'}}{b}$. As argued before, in this case the path $a' \to b \to \latePartner{b'} \to b'$ is temporal and contained in $S$.

    Case 3: $\leqAt{\earlyPartner{a'}}{b'}{a}$. In this case, the path $a' \to \earlyPartner{a'} \to a \to b'$ is temporal and contained in $S$.
\end{proof}
\fi

Note that we only require a single edge $e$ to have a small number of non-$e$-reverted edges to obtain a small spanner. Fortunately, the shifted matching graph $\ringShift{n}$ contains many edges $e$ for which all edges are $e$-reverted.

\iflong
\begin{lemma}
\fi
\ifshort
\begin{lemma}[$\star$]
\fi
\label{lem:ringshift-reverted}
    Let $n \in \N^+$ and $e = \undirEdge{a}{b}$ be an edge with label $0$ or $n-1$ in $\ringShift{n}$. Then $\malicious{e} = \emptyset$.
\end{lemma}

\iflong
\begin{proof}
We show the statement only for $\edgeLabele{e} = 0$, the proof for label $n-1$ follows analogously. Without loss of generality, let $a = a_0$ and $b = b_0$ and let $i,j \in \setExc{n}$. To show that $\undirEdge{a_i}{b_j}$ is $e$-reverted, we distinguish two cases.

Case 1: $i \le j$. By the structure of the shifted matching graph, we know $\earlyPartner{a_i} = \leqAt{b_i}{b_j}{a_0}$, and $\undirEdge{a_i}{b_j}$ is $e$-reverted.

Case 2: $i > j$. In this case, we have $\leqAt{a_i}{a_{j+1\xmod n}}{b_0} = \latePartner{b_j}$. Thus, $\undirEdge{a_i}{b_j}$ is $e$-reverted.
\end{proof}
\fi

This allows us to apply \Cref{thm:malicious-solution} to $\ringShift{n}$ to construct a linear-size bi-spanner.

\begin{corollary}\label{lem:ringshift-solution}
    For every $n \in \N^+$, the graph $\ringShift{n}$ has a bi-spanner of size
    $4n-4$.
\end{corollary}

Now we connect the concept of $e$-reverted edges to partial pivot-edges from \Cref{sec:pivotable}.

\iflong
\begin{lemma}
\fi
\ifshort
\begin{lemma}[$\star$]
\fi
\label{lem:malicious-pivot}
    Let $\defaultBiclique$ be a bi-clique and $e = \undirEdge{a}{b} \in \bicliqueEdges{A}{B}$. Then $\abs{\malicious{e}} \ge \abs{A \setminus \inSet{e}} \cdot \abs{B \setminus \outSet{e}}$.
\end{lemma}

\iflong
\begin{proof}
    Let $b' \in B \setminus \outSet{e}$ and $a' \in A \setminus \inSet{e}$. This also implies $\latePartner{b'} \notin \outSet{e}$ and $\earlyPartner{a'} \notin \inSet{e}$. Therefore, we have $\lessAt{\latePartner{b'}}{a}{b}$ and $\lessAt{a}{a'}{b}$ and by transitivity of $\lessOp{b}$ we get $\lessAt{\latePartner{b'}}{a'}{b}$. Additionally, $\lessAt{b'}{b}{a}$ and $\lessAt{b}{\earlyPartner{a'}}{a}$ follows from the definitions of $a'$ and $b'$. Again, using transitivity we get $\lessAt{b'}{\earlyPartner{a'}}{a}$ and $\undirEdge{a'}{b'} \in \malicious{e}$.
\end{proof}
\fi

\Cref{lem:malicious-pivot} tells us something about where to look for edges $e$ with many $e$-reverted edges. In fact, only edges that are either among the earliest or the latest edges of the incident vertices are suitable candidates for \Cref{thm:malicious-solution}, which are the least likely to be suitable partial pivot-edges. This gives us an intuition, why the two techniques can be applied in different situations. This intuition is supported by the example of shifted matching graphs.

% Keep this at the end of each file
%%% Local Variables:
%%% mode: latex
%%% TeX-master: "../writeup.tex"
%%% End:

\section{Product Graphs}\label{sec:swarms}
In this section we present a general construction to compose two bi-cliques.
We show that a careful application yields a bi-clique which cannot be solved by the techniques presented so far. 
In particular, we present a class of graphs, such that for any edge \(e\) the pivot-set of $e$ is sub-linear and the set of non-revertible vertex-pairs grows almost quadratically.

\begin{figure}[t]
    \centering

    \begin{subfigure}[t]{0.49\textwidth}
        \centering
        \includegraphics[width=63.25mm]{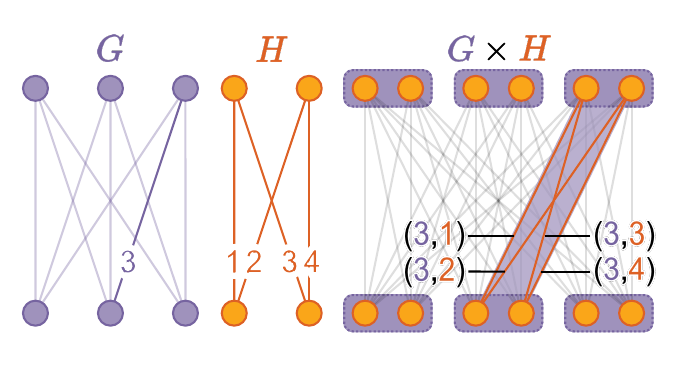}
        \caption{Construction of product graphs. Observe how every edge in $G$ is replaced by a copy of the graph $H$ and every label is the pair of their original labels.}
        \label{fig:product}
    \end{subfigure} %
    \hfill
    \begin{subfigure}[t]{0.49\textwidth}
        \centering
        \includegraphics[width=44.3mm]{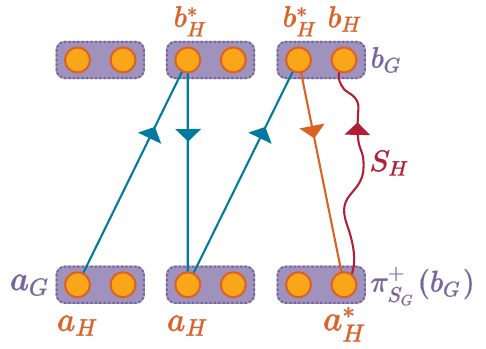}
        \caption{The construction of the path in the proof of \Cref{lem:general-product-spanner}. Purple labels are the first element of vertex names, orange ones are the second element.}
        \label{fig:product-solution}
    \end{subfigure}%
    \caption{Product graphs and how to find small spanners.}
\end{figure}

\iflong
\begin{theorem}
\fi
\ifshort
\begin{theorem}[$\star$]
\fi
\label{thm:small-pivot-large-malicious}
  Let $f\colon \N \to \N^{+}$ be any function with \(f(n) \in \bigO(n)\). 
  There is a set of bi-cliques \(\{D_{n}\}_{n\in \N}=\{(A_{n},B_{n},\lambda_{n})\}_{n\in \N}\), such that
  \begin{enumerate}
    \item \(|A_{n}|=|B_{n}|\in\Theta(n)\);
    \item for all \(e \in \bicliqueEdges{A_n}{B_n}\)
    \begin{enumerate}
      \item the size of \(\inSet{e} \cap \outSet{e}\) is in
$\bigO\left(f(n)\right)$,
      \item the number of the not \(e\)-reverted edges \(\malicious{e}\) is in
$\Omega\left(nf(n)\right)$.\vspace{2mm}
    \end{enumerate}
  \end{enumerate}
\end{theorem}

We do not prove \Cref{thm:small-pivot-large-malicious} right away. Rather, we
first consider an application. For this, choose \(f(n) \coloneqq n^{1-\varepsilon}\).
\Cref{thm:small-pivot-large-malicious} tells us that there is a set of graphs
\(\{D_{n}\}_{n\in \N}\) such that the size of sides is in \(\Theta(n)\), the size of any
pivot-set is in \(\bigO(n^{1-\varepsilon})\), and the number of not reverted edges with
respect to any edge is in \(\Omega(n^{2-\varepsilon})\).

To prove \Cref{thm:small-pivot-large-malicious}, we define a graph
class that we later use to define appropriate \(D_{n}\).

\begin{definition}[Product Graph]
    Consider two bi-cliques $G = (A_G, B_G, \edgeLabel_G)$
    and $H=(A_H, B_H, \edgeLabel_H)$. Define the
    \emph{product graph} $\exGraph{G}{H} = (\mathcal{A}, \mathcal{B}, \Lambda)$ as
    \begin{align*}
        \mathcal{A} &\coloneqq A_G \times A_H, \\
        \mathcal{B} &\coloneqq B_G \times B_H, \\
        \Lambda &\coloneqq ((a_G,a_H), (b_G,b_H)) \mapsto (\edgeLabel_G(a_G, b_G), \edgeLabel_H(a_H, b_H)). 
    \end{align*}
\end{definition}

Note that the co-domain of $\Lambda$ is $\N^2$. To make it the proper definition of an
edge labeling function, take the lexicographic embedding. For the reader's
benefit, we continue to write tuples. Intuitively, to construct \(\exGraph{G}{H}\), we start
with the graph \(G\) and replace each vertex of \(A_{G}\) with a distinct
copy of \(A_{H}\) as well as each vertex of \(B_{G}\) with a distinct copy of
\(B_{H}\), see \Cref{fig:product} for an illustration of this idea. We refer to each such expanded vertex as a \emph{bag}. Note that this operation is not commutative. 
The edge labels \(\Lambda\) are chosen in such a way that any 2-edge path \(x\to y \to z\) with
\(x\) and \(z\) in the same bag is temporal if the respective path in \(H\) is
temporal and any 2-edge path \(x\to y \to z\) with \(x\) and \(z\) belonging to
different bags is temporal if the respective path is temporal in \(G\).
We now make this rigorous.

\iflong
\begin{lemma}
\fi
\ifshort
\begin{lemma}[$\star$]
\fi
\label{lem:expanded-graph-contract-paths}
  Let $G$ and $H$ be
  bi-cliques. Let \(P=(g_{1}, h_{1})(g_{2},h_{2})\dots (g_{\ell},h_{\ell})\) be a temporal path in
\(\exGraph{G}{H}\). Then
  \begin{enumerate}
    \item the sequence \(g_{1}g_{2}\dots g_{\ell}\) is a temporal path in \(G\);
    \item if \(\edgeLabel_{G}(g_{1}, g_{2}) = \edgeLabel_{G}(g_{\ell-1},g_{\ell})\), the sequence \(h_{1}h_{2}\dots h_{\ell}\) is a
          temporal path in \(H\).
  \end{enumerate}
\end{lemma}

\iflong
\begin{proof}
  By choice of \(\Lambda\), for all \(i \in [\ell - 2]\) we have
  \(\edgeLabel_{G}(g_{i},g_{i+1}) \leq \edgeLabel_{G}(g_{i+1},g_{i+2})\) and so
  \(g_{1}g_{2}\dots g_{\ell}\) is temporal in \(G\).
  Additionally, if
\(\edgeLabel_{G}(g_{1}, g_{2}) = \edgeLabel_{G}(g_{\ell-1},g_{\ell})\), for all $i \in [\ell-1]$ the values
\(\edgeLabel_{G}(g_{i}, g_{i+1})\) are equal.
  By lexicographic ordering, for all \(i \in [\ell - 2]\) we have
  \(\edgeLabel_{H}(h_{i}, h_{i+1}) \leq \edgeLabel_{H}(h_{i+1}, h_{i+2})\) and so
  \(h_{1}h_{2}\dots h_{\ell}\) is temporal in \(H\).
\end{proof}
\fi

This gives us enough tools to prove how the pivot-set of an edge
\(e=\undirEdge{(u_{G},u_{H})}{(v_{G},v_{H})}\in E(\exGraph{G}{H})\) in \(\exGraph{G}{H}\) relates to the pivot-set of
\(\undirEdge{u_{G}}{v_{G}}\) in \(G\). For this, define
\(\bagMapOp\colon V(\exGraph{G}{H})\to V(G)\) as
\((u_{G},u_{H}) \overset{\bagMapOp}{\mapsto} u_{G}\).
This denotes the function that yields for every vertex in \(\exGraph{G}{H}\) its
bag, namely, its original vertex in \(G\). We will also
naturally extend \(\bagMapOp\) to edges \ie{\(\bagMap{\undirEdge{(u_{G}, u_{H})}{(v_{G},v_{H})}} = \undirEdge{u_{G}}{v_{G}}\)}. Intuitively, for any vertex
\(v\) in the pivot-set of an edge \(e\in E(\exGraph{G}{H})\), the vertex
\(\bagMap{v}\) must be in the pivot-set of \(\bagMap{e}\) in \(G\).

\iflong
\begin{lemma}
\fi
\ifshort
\begin{lemma}[$\star$]
\fi
\label{lem:expand-graph-in-out-set}
  Let \(G\) and \(H\) be bi-cliques and let \(e\in E(\exGraph{G}{H})\). We have \begin{enumerate}
      \item $\bagMap{\inSet{e}}\subseteq \inSet{\bagMap{e}}$,
      \item $\bagMap{\outSet{e}}\subseteq \outSet{\bagMap{e}}$, and
      \item $\bagMap{\inSet{e} \cap \outSet{e}} \subseteq \inSet{\bagMap{e}} \cap \outSet{\bagMap{e}}$.
  \end{enumerate}
\end{lemma}

\iflong
\begin{proof}
  The third statement follows from the first two by \[
  \bagMap{\inSet{e} \cap \outSet{e}} \subseteq \bagMap{\inSet{e}} \cap \bagMap{\outSet{e}} \subseteq \inSet{\bagMap{e}} \cap \outSet{\bagMap{e}}.
  \]

  We prove the first statement, the second is analogous. 
Consider any \(e=\undirEdge{u}{v}\in E(\exGraph{G}{H})\) and \(w \in \inSet{e}\). We need to show
that there is a temporal path from \(\bagMap{w}\) to one of \(\bagMap{u}\) or
\(\bagMap{v}\) arriving not later than \(\edgeLabel_{G}(\bagMap{e})\).  For this
consider the path \(P=p_{1}p_{2}\dots p_{\ell}\) from \(w\) to \(e\) that does not
end later than \(\Lambda(e)\). By \Cref{lem:expanded-graph-contract-paths}, the path
\(P'\coloneqq\bagMap{p_{1}}\bagMap{p_{2}}\dots\bagMap{p_{\ell}}\) is temporal in \(G\). Note
that, \(\bagMap{p_{1}} = \bagMap{w}\) and \(\bagMap{p_{\ell}}\in \set{\bagMap{u}, \bagMap{v}}\). For
the last edge \(e'=p_{\ell -1}p_{\ell}\) of \(P\), we have \(\Lambda(e') \leq \Lambda(e)\), and so
\(\edgeLabel_{G}(\bagMap{e'})\leq \edgeLabel_{G}(\bagMap{e})\). Thus, \(P'\) is a
path we were searching for.
\end{proof}
\fi

We now focus on product graphs, where both $G$ and $H$ are shifted matching graphs, as we will use these to prove
\Cref{thm:small-pivot-large-malicious}. Therefore, they will serve as the family of graphs which the already introduced techniques do not solve.

\begin{definition}
    Let $m,k \in \N$, Define the graph
    $\smsmbg{m}{k}=\exGraph{\ringShift{m}}{\ringShift{k}}$ to be the product graph
obtained by taking $\ringShift{m}$ as the outer graph and $\ringShift{k}$ as the
inner graph.
\end{definition}

Observe that each side of $\smsmbg{m}{k}$ has \(m\) bags each of size \(k\) and
so, overall, \(\smsmbg{m}{k}\) has \(mk\) vertices per side. We now investigate
the size of the pivot-sets in \(\smsmbg{m}{k}\) as well as the number of the
non-reverted edges.

\iflong
\begin{lemma}
\fi
\ifshort
\begin{lemma}[$\star$]
\fi
\label{lem:smsmbg-pivot}
  Let $m,k \in \N$ and consider $\smsmbg{m}{k}$. For any \(e \in E(\smsmbg{m}{k})\),
  we have that the size of the pivot-set \(\inSet{e}\cap\outSet{e}\) is at most \(2k\).
\end{lemma}

\iflong
\begin{proof}
  Let \(e=\undirEdge{u}{v}\in E(\smsmbg{m}{k})\). By
  \Cref{lem:ringshift-no-pivot-edge,lem:expand-graph-in-out-set}, we have
  \(\bagMap{\inSet{e}\cap\outSet{e}} \subseteq \set{\bagMap{u},\bagMap{v}}\) and therefore every $w \in \inSet{e}\cap\outSet{e}$ must be in the same bag as $u$ or $v$. As the size of each
  bag is \(k\), we have \(\abs{\inSet{e}\cap\outSet{e}}\leq 2k\).
\end{proof}
\fi

\iflong
\begin{lemma}
\fi
\ifshort
\begin{lemma}[$\star$]
\fi
\label{lem:smsmbg-malicous}
    Let $k,m \in \N$ and consider $\smsmbg{m}{k}$. For any
    \(e\in E(\smsmbg{m}{k})\), the size of the set of not \(e\)-reverted edges
    \(\malicious{e}\) is at least \((m - 1)\binom{k}{2}\).
\end{lemma}

\iflong
\begin{proof}
  \newcommand{\aBag}[1]{A_{#1}}
  \newcommand{\bBag}[1]{B_{#1}}
  Recall, that there is an order on the super vertices
  \(a_{0},a_{1},\dots, a_{m-1}\) and \(b_{0},b_{1},\dots, b_{m-1}\) such that
  for all \(a_{i}\in A_{G}\) and \(b_{j} \in B_{G}\) we have
  \(\edgeLabel_G(\undirEdge{a_{i}}{b_{j}}) = (j-i)\xmod m\). For all \(0\leq i< m\), we write
\(\aBag{i} \coloneqq \bagMapOp^{-1}(a_{i})\) and
\(\bBag{j} \coloneqq \bagMapOp^{-1}(b_{i})\).

Now consider any edge \(e \in E(\smsmbg{m}{k})\). Notice that for all \(t \in \Z\), the
temporal graph \(\smsmbg{m}{k}\) is isomorphic to itself (in the sense that the map does not change the labels between two vertices) with the function
\[
  (a_i, a_j) \mapsto 
            (a_{(i + t)\xmod m}, a_j) \text{ and } (b_i, b_j) \mapsto 
            (b_{(i + t)\xmod m}, b_j).
\] We can therefore assume without loss of generality, that \(e = \undirEdge{a}{b}\)
with \(a \in \aBag{0}\) and let \(j\) be such that \(b \in \bBag{j}\). Therefore, for any
\(0 \leq x < y<m\), \(b_{x}\in\bBag{x}\), and
\(b_y\in\bBag{y}\), we have \(\lessAt{b_{x}}{b_{y}}{a}\).

One can verify, that \(\smsmbg{m}{k}\) is extremally matched. So, we obtain that
\begin{align*}
  \malicious{e}&=
\set{\undirEdge{a'}{ \latePartner{a''}} \in E(\smsmbg{m}{k})}{\undirEdge{a'}{\latePartner{a''}}\text{ is not \(e\)-reverted}}\\
	    &=\set{\undirEdge{a'}{\latePartner{a''}} \in E(\smsmbg{m}{k})}{\greaterAt{a'}{a''}{b}\text{ and
    } \greaterAt{\earlyPartner{a'}}{\latePartner{a''}}{a}}.
\end{align*}

Now consider any \(1 \le \ell \le m-1\) as well as \(a',a''\in \aBag{\ell}\) with \(\greaterAt{a'}{a''}{b}\).
By the properties of shifted matching graphs, \(\earlyPartner{a'} \in \bBag{\ell}\) and
\(\latePartner{a''}\in \bBag{(\ell - 1)\xmod m} = \bBag{\ell-1}\). By choice of \(a\), we have
\(\greaterAt{\earlyPartner{a'}}{\latePartner{a''}}{a}\). So,
\((a', \latePartner{a''}) \in \malicious{e}\).

For every \(1 \le \ell \le m-1\), there are \(\binom{k}{2}\) distinct pairs contributing to
\(\malicious{e}\). Additionally, for every of the \(m-1\) choice of \(\ell\), these pairs are
different. Overall, \(\abs{\malicious{e}}\geq (m-1) \binom{k}{2}\).
\undef\aBag
\undef\bBag
\end{proof}
\fi

\iflong
Now we can finally conclude that these two lemmas together mean that our tools are not strong enough to solve this graph class yet.

\begin{proof}[Proof of \Cref{thm:small-pivot-large-malicious}]
  Let \(m_{n} \coloneqq 1+\left\lceil\frac{n}{f(n)}\right\rceil\) and
  \(k_{n}\coloneqq 1+f(n)\) and choose
\(G_{n} = \smsmbg{m_{n}}{k_{n}}\). Denote with
\(s_{n} = m_{n} k_{n}\) the size of each
side of \(G_{n}\). We have
\begin{align*}
s_{n}&= m_{n}k_{n}=\left(1+\left\lceil\frac{n}{f(n)}\right\rceil\right)\left(1 + f(n)\right)
  \geq \frac{n}{f(n)}\cdot f(n) \geq n,
\intertext{and}
s_{n}&\leq \left(2 + \frac{n}{f(n)}\right)\cdot (1 + f(n))\leq 2+ 2f(n) + 2n.
\intertext{
Since \(f(n)\in \bigO(n)\), we have
\(s_{n} \in \Theta(n)\). By \Cref{lem:smsmbg-pivot}, we have that any pivot-set is in
\(\bigO(k_{n})=\bigO(f(n))\), and by \Cref{lem:smsmbg-malicous} any set of not reverted edges
has size at least}
  (m_{n} - 1) \binom{k_{n}}{2}&\geq \frac{n}{f(n)}\left(\frac{(f(n) + 1)\cdot f(n)}{2}\right)
  \geq\frac{nf(n)}{2}
.\qedhere
\end{align*}
\end{proof}
\fi

\ifshort
The last two lemmas show that our tools are not strong enough to solve this graph class yet.
\fi
With our previous techniques, we are not able to give a linear-size spanner for every
product graph. By \Cref{thm:small-pivot-large-malicious}, for $f(n) \in
\Theta(\sqrt{n})$, \Cref{thm:edge-pivot} and
\Cref{thm:malicious-solution} only give us $\Theta(n^{3/2})$ spanners. Still,
product graphs exhibit a lot of structure, in the sense that all edges between
two bags have roughly the same time labels. We use this to construct bi-spanners
for the product graph based on bi-spanners for the two underlying graphs.

\iflong
\begin{lemma}
    \label{lem:general-product-spanner}
Let $G = \biclique{A_G}{B_G}{\edgeLabel_G}$ and $H = \biclique{A_H}{B_H}{\edgeLabel_H}$ be temporal bi-cliques of sizes $n_G$ and $n_H$ and let $S_G, S_H$ be temporal bi-spanners respectively. Then, there is a bi-spanner in $\exGraph{G}{H}$ of size at most $|S_G|n_H + |S_H|n_G$.
\fi
\ifshort
\begin{lemma}[$\star$]
    \label{lem:general-product-spanner}
Let $G = \biclique{A_G}{B_G}{\edgeLabel_G}$ and $H = \biclique{A_H}{B_H}{\edgeLabel_H}$ be temporal bi-cliques of sizes $n_G$ and $n_H$ and let $S_G, S_H$ be temporal bi-spanners respectively. Then, there is a bi-spanner in $\exGraph{G}{H}$ of size at most $|S_G|\cdot |S_H|$.
\fi
\end{lemma}

\newcommand{\lateBag}[1]{\lateMatchingG{S_G}({#1})}
\iflong
\begin{proof}
\newcommand{\rb}[1]{b_H^*} %r_B({#1})}
\newcommand{\ra}[1]{a_H^*} %r_A({#1})}

  Let $b_G \in B_G$. Define $\lateBag{b_G}$ to be the $a_G \in A_G$, such that
  $\undirEdge{a_G}{b_G}$ is the latest edge in $S_G$ that is incident to $b_G$.
  Note that since $S_G$ is a spanner, every $b_G$ must be incident to at least
  one edge. Furthermore, choose $b_H^* \in B_H$ and $a_H^* \in A_H$ such that
  the label of $\undirEdge{a_H^*}{b_H^*}$ is minimal. 
  
  We construct a temporal spanner $S$ for $\exGraph{G}{H}$ as follows. For every
  $\undirEdge{a_G}{ b_G} \in S_G$ and $a_H \in A_H$, add $\undirEdge{(a_G,
  a_H)}{(b_G, \rb{b_G})}$ to $S$. Additionally, for every $\undirEdge{a_H}{b_H}
  \in S_H$ and $b_G \in B_G$, add $\undirEdge{(\lateBag{b_G}, a_H)}{(b_G, b_H)}$
  to $S$.

  We can immediately verify that this results in a set of size $|S| \leq
  |S_G|n_H + |S_H|n_G$. To show temporal connectivity, let $(a_G, a_H) \in A_G
  \times A_H$ and $(b_G, b_H) \in B_G \times B_H$, we aim to construct a path
  $P$ from $(a_G, a_H)$ to $(b_G, b_H)$. Since $S_G$ is a bi-spanner for $G$,
  there is a temporal path $P_G$ from $a_G$ to $b_G$ in $S_G$. For every edge
  $\undirEdge{a_G^P}{b_G^P}$ in $P_G$ except for the last one, we append
  $\undirEdge{(a_G^P, a_H)}{(b_G^P, \rb{b_G^P})}$ to $P$. Note that these edges
  are in $S$, since $\undirEdge{a_G^P}{b_G^P} \in S_G$. Furthermore, it is easy
  to verify that two neighboring edges share an endpoint and since $P_G$ is
  temporal, path $P$ is temporal in $\exGraph{G}{H}$. Since we left out exactly
  the last edge of $P_G$, the path $P$ ends in some vertex $(a_G', a_H)$. We
  distinguish two cases. 
  
  Case 1: $a_G' = \lateBag{b_G}$. As $S_H$ is a bi-spanner for $H$, there exists
  a temporal path $P_H$ from $a_H$ to $b_H$ in $S_H$. For an edge
  $\undirEdge{a_H^P}{b_H^P}$ in $P_H$, append $\undirEdge{(a_G', a_H^P)}{(b_G,
  b_H^P)}$ to $P$. Since $a_G' = \lateBag{b_G}$, this edge is in $S$ by
  construction. Furthermore, these edges have a higher label than all previous
  edges in $P$, since $\undirEdge{a_G'}{b_G}$ must be the last edge of $P_G$. As
  $P_H$ is temporal, path $P$ remains temporal and it is easy to verify that $P$
  is indeed a path from $(a_G, a_H)$ to $(b_G, b_H)$.

  Case 2: $a_G' \ne \lateBag{b_G}$. We append $\undirEdge{(a_G', a_H)}{(b_G,
  \rb{b_G})}$, which is in $S$ because $\undirEdge{a_G'}{b_G}$ is the last edge
  of $P_G$. Further, we can append $\undirEdge{(\lateBag{b_G}, \ra{b_G})}{(b_G,
  \rb{b_G}}$ to $P$, since $\undirEdge{\lateBag{b_G}}{b_G} \in S_G$. Again, as
  $S_H$ is a bi-spanner for $H$, there exists a temporal path $P'_H$ from
  $\ra{b_G}$ to $b_H$ in $S_H$. In the same fashion as in Case 1, for every edge
  $\undirEdge{a_H^P}{b_H^P}$ in $P'_H$, we append $\undirEdge{(\lateBag{b_G},
  a_H^P)}{(b_G, b_H^P)}$ to $P$. See \Cref{fig:product-solution} for a
  visualization of the constructed path $P$. As $\undirEdge{\ra{b_G}}{\rb{b_G}}$
  has the minimal label in $S_H$ and $a_G' \ne \lateBag{b_G}$, our construction
  yields a temporal path which is clearly included in $S$.
\end{proof} 
\fi

\ifshort
With more careful analysis, this construction can be improved to a bound of $|S_G|n_H + |S_H|n_G$.
\fi

If $G$ and $H$ each have linear-size bi-spanners, this proves the existence of a linear-size bi-spanner for $\exGraph{G}{H}$.
Since we know from \Cref{sec:non-pivotable} how to construct a linear-size bi-spanner for $\ringShift{m}$, we can now construct a linear-size bi-spanner for $\smsmbg{m}{k}$.

Additionally, we can augment \Cref{lem:general-product-spanner} to also deal with different graphs $H_e$ per edge~$e$ in $G$. 
To do this, we have to change the definition of our representatives from $a_H^*$ and $b_H^*$ to functions $r_A, r_B$, such that for $b_G \in B_G$ 
the label of $\undirEdge{(\lateBag{b_G}, r_A(b_G))}{(b_G, r_B(b_G))}$ is minimal.
This enables us not only to construct bi-spanners for $G \times H$ but any bi-cliques that can be decomposed into bags of vertices.

% Keep this at the end of each file
%%% Local Variables:
%%% mode: latex
%%% TeX-master: "../writeup.tex"
%%% End:

\section{Conclusion}
\begin{figure}[t]
    \centering
    \includegraphics[width=56mm]{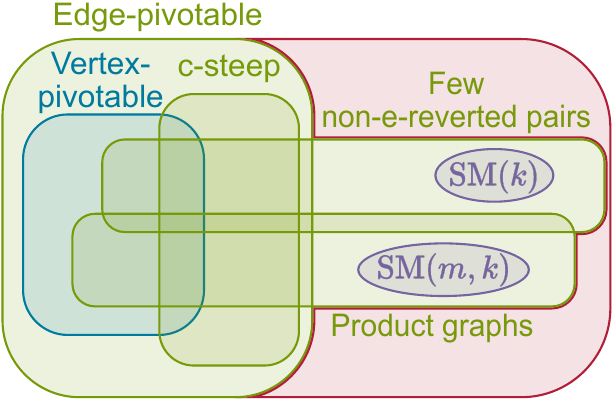}
    \caption{Overview of techniques. Previously proposed techniques are blue, new ones are green. Example families of graphs are purple. The red area contains instances that are not solved by any known techniques.}
    \label{fig:overview}
\end{figure}

We present a significant step forwards towards answering whether temporal
cliques admit linear-size temporal spanners, see \Cref{fig:overview} for an
illustration. The obvious problem that stems from this paper, is whether classes
of graphs exist that are not solved by our techniques. We conjecture that this
is the case but the classes of graphs remaining should be quite artificial.
Another possible avenue is to generalize some of our techniques to include a
wider variety of graphs. For example, we could consider investigating the
temporal ordering of edges and try to generalize the concept of reverted edges.
Also, product graphs and their solution suggest a way in which temporal
bi-cliques may be composed and decomposed. Future research could generalize this
idea and more closely analyze the cases in which a temporal spanner can be
obtained by combining temporal spanners on subinstances. Towards the opposite
direction, it is now easier to find a $\Omega(n \log n)$ counterexample, if it
exists at all, since we have provided many properties that any counterexample must avoid.

\bibliography{paper}

\end{document}